\documentclass{amsart}
\usepackage{IEEEtrantools}

\usepackage{graphicx} 
\usepackage{amsmath,amsthm,amssymb} 
\usepackage[usenames]{color}

\usepackage{tikz}
\usetikzlibrary{automata}
\usetikzlibrary{shapes}
\usetikzlibrary{matrix} 
\usetikzlibrary{arrows} 
\usetikzlibrary{calc} 

\newtheorem{theorem}{Theorem}
\newtheorem{lemma}{Lemma}
\newtheorem{corollary}{Corollary}
\newtheorem{definition}{Definition}
\newtheorem{example}{Example}
\newtheorem{remark}{Remark}

\newcommand{\B}{\mathbb{B}}
\newcommand{\R}{\mathbb{R}}
\newcommand{\N}{\mathbb{N}}
\newcommand{\Z}{\mathbb{Z}}

\newcommand{\sd}[1]{\mathsf{d}_\mathsf{#1}}

\renewcommand{\emptyset}{{\varnothing}}

\newcommand{\intcc}[1]{\ensuremath{{\left[#1\right]}}}
\newcommand{\intoc}[1]{\ensuremath{{\left]#1\right]}}}
\newcommand{\intco}[1]{\ensuremath{{\left[#1\right[}}}
\newcommand{\intoo}[1]{\ensuremath{{\left]#1\right[}}}

\DeclareMathOperator{\dom}{dom} 
\DeclareMathOperator{\defeq}{\mathrel{\mathop:}=}

\begin{document}

\title{A Notion of Robustness \\ for Cyber-Physical Systems}

\author{Matthias Rungger}
\address{Department of Electrical Engineering, University of California at Los Angeles, CA 90095-1594, United States}
\email{rungger@ucla.edu}

\author{Paulo Tabuada}
\address{Department of Electrical Engineering, University of California at Los Angeles, CA 90095-1594, United States}
\email{tabuada@ee.ucla.edu}

\date{}

\begin{abstract}
Robustness as a system property describes the degree to which a system
is able to function correctly in the presence of disturbances, i.e.,
unforeseen or erroneous inputs. In this paper, we introduce a notion
of robustness termed \emph{input-output dynamical stability} for
cyber-physical systems (CPS) which merges existing notions of
robustness for continuous systems and discrete systems. The notion
captures two intuitive aims of robustness: bounded disturbances have
bounded effects and the consequences of a sporadic disturbance
disappear over time. We present a design methodology for robust CPS
which is based on an abstraction and refinement process. We suggest
several novel notions of simulation relations to ensure the soundness
of the approach. In addition, we show how such simulation relations
can be constructed compositionally. The different concepts and
results are illustrated throughout the paper with examples.
\end{abstract}
\maketitle

\markboth{A NOTION OF ROBUSTNESS FOR CYBER-PHYSICAL SYSTEMS}{MATTHIAS RUNGGER AND PAULO TABUADA}

\section{Introduction}

Robustness describes the ability of a system to  function correctly
in the presence of disturbances, e.g., unmodeled dynamics or
unforeseen events. Disturbances arise whenever certain assumptions
imposed on the system or the environment at design-time
are violated during run-time. Since a system and its environment are only
partly known at design-time, disturbances are unavoidable and
robustness of is natural requirement in every system design.

In this paper we present a methodology for the design of robust
Cyber-Physical Systems (CPS). We establish robustness with respect to
continuous disturbances, possibly arising from sensor noise or
actuator imprecisions, as well as discrete disturbances to account for
potential failures of the cyber components, like faulty communication
channels, hardware or software errors.

Technically, we formalize robustness as \emph{input-output dynamical stability} as a formal 
notion of robustness of CPS, combining well-known notions of
robustness for control systems, such as input-to-state stability~\cite{Son08} and
input-to-state dynamical stability~\cite{Gru02}, with a recently
introduced notion of robustness for discrete
systems~\cite{TCRM14,TBCSM12}. Input-output dynamical stability
provides two guarantees that are intuitively related with a robust design: first, bounded
disturbances have bounded consequences and second,
the nominal system behavior is eventually resumed after the occurrence
of a sporadic disturbance.

We provide a computational framework based on a three-step abstraction
and refinement procedure. The first step, consists of computing a \emph{discrete
abstraction} or \emph{symbolic model}, i.e., a finite-state substitute
of a given CPS.  In the second step, we employ the algorithms
developed in~\cite{TCRM14,TBCSM12} to synthesize a robust controller
for the symbolic model. The last step, consist in the refinement
of the controller obtained on the abstract domain to the concrete CPS.

We follow the usual approach, which is based on \emph{simulation
relations} and \emph{alternating simulation relations}, to ensure the
soundness of the abstraction and refinement scheme. Simulation
relations provide a mathematical tool to compare the dynamical behavior of
a concrete system and its symbolic model in terms of behavioral inclusion.
In this paper, we enrich the well-known constructs of (alternating) simulation
relations~\cite{Mil89,AHKM98,GP07,Tab08} to facilitate the comparison of
two systems in terms of robustness.

We recently introduced in~\cite{RT13} \emph{contractive simulation
relations} to capture a certain stability or contraction property that
is often observed in the concrete system~\cite{Tab08,PGT08,PT09,GPT10}
with the goal to reduce the complexity of the symbolic models.
The focus of~\cite{RT13} was the verification of robustness using
contractive simulation relations. In this paper we focus on the
synthesis of robust controllers which naturally leads to the notion of
\emph{contractive Alternating Simulation Relations} (ASR). By using contractive
ASR we are allowed to \emph{ignore}
continuous disturbances on the abstract domain, while still providing
robustness of the CPS with respect to continuous as well as discrete
disturbances. As we will illustrate with an example in
Section~\ref{s:app}, this might lead to a \emph{separation of
concerns}, where a continuous design caters to continuous disturbances
 and the discrete design on the abstract domain caters to
discrete disturbances. Yet, the refined design provides robustness
with respect to both continuous and discrete disturbances.

While it is straightforward to construct symbolic models together with
contractive ASR for continuous control
systems following the methods presented
in~\cite{Tab08,PGT08,PT09,GPT10}, it is less clear how to construct
such models and relations for CPS. In Section~\ref{s:composition} we
provide a compositional scheme. This approach is in particular useful for
CPS, since the overall symbolic model of the CPS can be constructed
from the individual symbolic models of the physical part and the cyber
part of the CPS.

In summary, the contribution in this paper as follows: 1)
we introduce input-output dynamical stability as formal notion of robustness of
CPS and propose an abstraction/refinement scheme for the synthesis of
robust controllers; 2) we show how
to refine a robust design found on the abstract domain to a robust
design on the concrete domain whenever the symbolic model is related
to the concrete system by an ASR; 3) when
using contractive ASR we tailor the design on the abstract domain to
discrete disturbances, while ensuring the robustness of the refined
design with respect to continuous and discrete disturbances; 4) we
provide a compositional scheme to the construction of symbolic
models of CPS.

\subsection{Related work}

Robustness has been studied in the control systems community for more
than fifty years, see~\cite{Zam96}, and formalized in many different
ways including  operator finite gains, bounded-input bounded-output
stability, input-to-state stability, input-output stability, and
several others, see e.g. \cite{Son08}. Moreover, robustness
investigations have been conducted for different system models such as
continuous-time systems, 
sampled-data systems, networked control systems,
and general hybrid systems \cite{NTS99,NT04a,CT09,San14}. The notion of robustness
described in this paper benefited from all this prior work and was
directly inspired by input-to-state stability~\cite{Son08} and its
quantitative version: input-to-state dynamical stability~\cite{Gru02}.
Unlike the framework presented in this paper, most of the existing research on robustness of
nonlinear control systems does not consider constructive
procedures for the verification and controller synthesis enforcing
robustness. The only
exceptions known to the authors are \cite{GSP06,ZD13,HJND05}.
Unfortunately, the finite-state models that are used in those
approaches represent \emph{approximations} of the concrete dynamics,
rather than abstractions. Hence, the soundness of those methods is
not ensured.

Robustness for discrete systems also has a long standing history. For
example, Dijkstra's notion of self-stabilizing algorithms
in the context of distributed systems~\cite{Dij74}
requires the ``nominal'' behavior of the system to be resumed in
finitely many steps after the occurrence of a disturbance. As
explained in~\cite{TCRM14,TBCSM12}, self-stabilizing systems are a special
case of robust systems, as defined in this paper. In addition to
self-stabilization, there exist several different notions of
robustness for discrete systems. For example, in~\cite{SR13} a
systematic literature review is presented, where the authors distill
and categorize more than 9000 papers on software robustness.  In the
following, we focus on the few approaches that provide quantitative
measures of robustness for discrete systems and thereby are close to the framework
presented in this paper.

Let us first mention two notions of robustness for systems over finite
alphabets~\cite{TMD08} and reactive systems~\cite{BGHJ09} that we
think are the closest to the definition of robustness discussed in
this paper. Similarly to our methodology, the deviation of the system
behavior from its ``nominal'' behavior as well as the disturbances are
quantified. A system is said to be robust if its deviation from the
``nominal'' behavior is proportional to the disturbance causing that
deviation.  Although, this requirement captures the first intuitive
goal of robustness, those definitions do not require that the effect
of a sporadic disturbance disappears over time. See~\cite{TCRM14,TBCSM12} for
a more rigorous comparison of the robustness definitions.

Note that the work in~\cite{BGHJ09} on reactive systems demonstrates
how to quantify disturbances and their effects on the system
behavior in order to characterize safety specifications in terms of
robustness inequalities. However, it is unclear how to quantify
disturbances and their effects in order to encode liveness
specifications.
Some possible notions are given in~\cite{BCGHJ10,Ehl11,TOLM12}, where the
robustness of a system is expressed as the ratio of the
number of assumptions and guarantees the system meets. Those notions
of robustness are incompatible with our definition of robustness, and
further work is needed if we would like to express liveness
specifications through the notion of robustness presented in this
paper.

There exist different studies that characterize the robustness of
discrete systems in terms of a Lyapunov function, as it is done in~\cite{CL97,PMA94} for
discrete event systems, or in~\cite{MRT13} for $\omega$-regular automata and
in~\cite{RMF13} for software programs. Note that Lyapunov functions
represent a tool to establish robustness inequalities, but do not
provide a direct quantification of the effect of
disturbances on the system behavior.  Hence, further work is needed to
related Lyapunov functions, like those presented
in~\cite{CL97,PMA94,MRT13,RMF13}, to a robustness inequality that directly quantifies the consequences of 
disturbances on the system behavior.

Another interesting method to characterize robustness for programs is
outlined in~\cite{MS09} and~\cite{CGL12}. Programs are interpreted as
function that map input data to output data. A program is said to be
robust if the associated input-output function is continuous. In
comparison to our approach, in~\cite{MS09,CGL12} a program is assumed
to terminate on all inputs and is interpreted as a static function
while we consider CPS whose executions are non-terminating.

A preliminary version of this contribution appears in~\cite{RT14}
where we announce the main results presented in this paper. In
comparison to~\cite{RT14}, we provide detailed proofs of all
statements. Moreover, the result on the compositional construction of
contractive alternating simulation relations presented in Section~\ref{s:composition} is new.

\section{Preliminaries}

We denote by $\N=\{0,1,2,\ldots\}$ the set of natural numbers and by $\B_x(r)$  the closed ball centered at $x\in\R^n$ with radius $r\in\R_{\ge
0}$. We identify $\B(r)$ with $\B_0(r)$. We use $|x|$ and $|x|_2$ to denote the
$\infty$-norm and two-norm of $x\in\R^n$, respectively. Given
$x\in\R^n$ and $A\subseteq \R^n$, we use $|x|_{A}\defeq\inf_{x'\in
A}|x-x'|_2$ to denote the  Euclidean
distance between $x$ and $A$. 
Given a set $A\subseteq\R^n$ we use $[A]_\eta:=\{x\in A\mid \exists k\in\Z^n:
x=2k\eta\}$ to denote a uniform grid in $A$.
For $a,b\in\R$ with $a\le b$, we denote the closed, open, and half-open intervals in $\R$ by $\intcc{a,b}$,
$\intoo{a,b}$, $\intco{a,b}$ and $\intoc{a,b}$, respectively. For $a,b\in\Z$, $a\le b$ we
use $\intcc{a;b}$, $\intoo{a;b}$, $\intco{a;b}$ and $\intoc{a;b}$, to
denote the corresponding intervals in $\Z$.

Given a function $f:A\to B$ and $A'\subseteq A$ we use
$f(A'):=\{f(a)\in B\mid a\in A'\}$ to denote the image of $A'$ under $f$.
A set-valued function or mapping $f$ from $X$ to $Y$ is denoted by
$f:X\rightrightarrows Y$. Its domain is defined by $\dom f:=\{x\in
X\mid f(x)\neq\emptyset\}$. Given a sequence $a:\N\to A$ in some set
$A$, we use $a_t$ to denote its $t$-th
element and $a_{\intcc{0;t}}$ to denote its restriction to the
interval $\intcc{0;t}$.
The set of all finite sequences is denoted by $A^*$.
The set of all infinite sequences is denoted by $A^\omega$ and we
think of elements $a\in A^\omega$ as sequences $a:\N\to
A$.
Given a relation $R\subseteq A\times B$ we use $\pi_A(R)$ and
$\pi_B(R)$ to denote
its projection onto the set $A$ and $B$, respectively. 

We use the following classes of comparison functions:
\begin{itemize}
  \item ${\mathcal K}\defeq\{\alpha:\R_{\ge0}\to \R_{\ge0}\mid\alpha$ is continuous and strictly increasing with $\alpha(0)=0\}$
  \item ${\mathcal L}\defeq\{\alpha:\N\to \R_{\ge0}\mid\alpha$ is strictly decreasing with $\lim_{t\to\infty}\alpha(t)=0\}$
  \item ${\mathcal{KL}}\defeq\{\beta: \R_{\ge0}\times\N\to \R_{\ge0}\mid\forall t\in\N: \beta(\cdot,t)\in {\mathcal K}\text{ and }\forall c\in\R_{\ge0}: \beta(c,\cdot)\in {\mathcal L}\}$
  \item ${\mathcal{KLD}}\defeq\{\beta\in{\mathcal{KL}}\mid \forall c\in\R_{\ge0},\forall s,t\in\N:\mu(c,0)=c\land \mu(c,s+t)=\mu(\mu(c,s),t)\}$
\end{itemize}
Note that we  work only  with discrete-time systems and for this reason we have defined
the domain of class ${\mathcal L}$ functions as~$\N$.

\section{Robustness for CPS}
\label{s:cpsys}

Since CPS exhibit a rich dynamical
behavior through the interaction of discrete and continuous components
we need an adequate mathematical description that is able to
represent its complex dynamics. We use a general notion of
transition system as the underlying model of CPS.

\begin{definition}
A \emph{system} $S$ is a tuple $S=(X,X_0,U,r)$
consisting of
\begin{itemize}
  \item a set of states $X$;
  \item a set of initial states $X_0\subseteq X$;
  \item a set of inputs $U$ containing the distinguished symbol $\bot$;
  \item a transition map $r:X\times U\rightrightarrows X$.
\end{itemize}

A \emph{behavior} of $S$ is a pair of sequences $(\xi,\nu) \in
(X\times U)^\omega$, that satisfies
$\xi_0\in X_0$ and $\xi_{t+1}\in r(\xi_t,\nu_t)$ for all times $t\in\N$.

A state $x\in X$ is called \emph{reachable} if there exists $T\in\N$
and sequences $\xi\in X^T$, $\nu\in U^{T-1}$ with 
$\xi_{t+1}\in r(\xi_t,\nu_t)$ for all \mbox{$t\in\intco{0;T}$}, $\xi_0\in
X_0$, and  $\xi_T=x$.

A system is
called \emph{non-blocking} if
$r(x,u)\neq\emptyset$ for any reachable state $x$ and any $u\in U$. It is
called \emph{finite} if $X$ and $U$ are finite sets and otherwise it
is called \emph{infinite}.
\end{definition}

Behaviors are defined as infinite sequences since we have in mind
reactive systems, such as control systems, that are required to
interact with its environment for arbitrarily long periods of time. In
particular, we are interested in understanding the effect of
disturbances on the system behavior. Therefore, the inputs in $U$ are
to be interpreted as disturbance inputs. Nevertheless, in order to
allow for the possibility of absence of disturbances, we assume that
$U$ contains a special symbol $\bot\in U$ that indicates that no
disturbance is present.

For simplicity of presentation, we assume throughout this section that the system is
non-blocking, i.e., for every state and (disturbance) input there exists at least
one successor state to which the system can transition.

In order to be able to talk about robustness properties, we endow our notion of system
with cost functions $I$ and $O$ that we use to describe the desired behavior and to
quantify disturbances. 

\begin{definition}
A \emph{system with cost functions} is a triple $(S,I,O)$ where $S$ is
a system and \mbox{$I:X\times U\to \R_{\ge0}$} and $O:X\times
U\to\R_{\ge0}$ are the \emph{input cost function} and
\emph{output cost function}, respectively.
\end{definition}

We now introduce a notion of robustness following well-known notions of
robustness for control systems, see e.g.~\cite{Son08}. In particular,
we follow the notion of \emph{input-to-state dynamical stability} introduced
in~\cite{Gru02} and generalize it here to CPS using the cost functions $I$ and $O$. 

\begin{definition}\label{d:iods}
Let $(S,I,O)$ be a system with cost functions, $\gamma\in{\mathcal K}$,
$\mu\in{\mathcal{KLD}}$ and $\rho\in\R_{\ge0}$. We say that $S$ is
\emph{$(\gamma,\mu,\rho)$-practically input-output
dynamically stable ($(\gamma,\mu,\rho)$-pIODS) with respect to
$(I,O)$} or that \emph{$(S,I,O)$ is $(\gamma,\mu,\rho)$-pIODS} if the following
inequality holds for every behavior of~$S$:
\begin{IEEEeqnarray}{C'C}\label{e:iods}
  O(\xi_{t},\nu_t)\le\max_{t'\in\intcc{0;t}}\mu(\gamma(I(\xi_{t'},\nu_{t'})),t-t')+\rho,& \forall t\in\N.
\end{IEEEeqnarray}
We say that \emph{$(S,I,O)$ is pIODS} if there exist
$\gamma\in{\mathcal K}$, $\mu\in{\mathcal{KLD}}$ and $\rho\in\R_{\ge0}$ such that
$(S,I,O)$ is $(\gamma,\mu,\rho)$-pIODS.

We say that \emph{$(S,I,O)$ is $(\gamma,\mu)$-IODS} if it is
$(\gamma,\mu,0)$-pIODS, and \emph{IODS} if there exist
$\gamma\in{\mathcal K}$, $\mu\in{\mathcal{KLD}}$ such that
$(S,I,O)$ is $(\gamma,\mu)$-IODS.
\end{definition}

If the cost functions are clear from the context or are irrelevant to the
discussion, we abuse the terminology and call a system $S$
pIODS/IODS without referring to the cost functions.

In our previous work~\cite{TBCSM12,TCRM14} we used IODS as a notion of
robustness for cyber systems. The underlying model were
\emph{transducers}, i.e., maps $f:U^*\to Y^*$ that process input streams in $U^*$ into output
streams in $Y^*$. In that framework, the cost functions were defined on
\emph{sequences} of input symbols and output symbols, i.e.,
$I:U^*\to\N$ and $O:Y^*\to\N$. In order formulate such cost functions
in the current framework we can compose the transducers computing the
input and output costs with the system being modeled so that input and
output costs are readily available as functions on the states and
inputs of the composed system.

Let us describe how the IODS inequality~\eqref{e:iods} realizes the
intuitive notion of robustness described in the introduction.
For the following discussion, suppose we are given a
system with cost functions $(S,I,O)$ that
is $(\gamma,\mu)$-IODS.
We use 
the output cost to specify preferences on the system behaviors: less
preferred behaviors have higher costs. In particular, the cost should
be zero for the nominal behavior. Similarly, we use the input costs to quantify the
disturbances. Hence, the input costs should be zero if no
disturbances are present, i.e., $I(\xi_t,\nu_t)=0$
when $\nu=\bot^\omega$. Since, $\gamma(0)=0$ and
$\mu(0,s)=0$ for all $s\in \N$, zero input cost implies zero output
cost which, in turn, implies that the system follows the desired behavior.
Moreover, inequality~\eqref{e:iods} implies that bounded disturbances lead to bounded
deviations from the nominal behavior. Suppose
$I(\xi_t,\nu_t)\le c$ holds for some $c\in\R_{\ge0}$ for all $t\in\N$.
Note that $\gamma$ is  monotonically increasing and $\mu(c,t)\le \mu(c,0)=c$ holds for all $t\in\N$.  Therefore,
\eqref{e:iods} becomes
\begin{IEEEeqnarray*}{C'C}
  O(\xi_{t},\nu_t)\le
  \gamma(c)& \forall t\in\N.
\end{IEEEeqnarray*}
In addition, inequality~\eqref{e:iods} ensures that the effect of
a sporadic disturbance vanishes over time. Suppose there exists $t'\in\N$ after
which the input cost is zero, i.e., $I(\xi_t,\nu_t)=0$ for all
$t\ge t'$. Then it follows from the definition of
$\mu\in{\mathcal{KLD}}$
that 
\begin{IEEEeqnarray*}{C'C}
\mu(\gamma(I(\xi_{t'},\nu_{t'})),t-t')\to 0,& t\to\infty.
\end{IEEEeqnarray*}
Hence, the output cost is forced to decrease to zero as time progresses.

We refer the reader to our
previous work~\cite{TBCSM12,TCRM14}
for a further demonstration of the usefulness of
inequality~\eqref{e:iods} to express robustness of cyber systems.
We showed in~\cite{TBCSM12,TCRM14} that verifying if a cyber system is robust
can be algorithmically solved in polynomial time. Similarly, the problem
of synthesizing a controller to enforce robustness of a cyber system
is solvable in polynomial time. Moreover, we provided some examples of robust cyber 
systems in the sense of inequality~\eqref{e:iods}.

\section{Preservation of IODS by Simulation Relations}

In this section we introduce simulation relations between two systems and answer the following question:
\begin{center}
  \emph{Under what conditions is pIODS preserved by simulation relations?} 
\end{center}
We consider three different types of relations: exact simulation
relations (SR), approximate simulation relations (aSR) and
approximate contractive simulation relations (acSR).

Bisimilarity and
(bi)simulation relations were introduced in computer science by
Milner and Park in the early 1980s, see e.g.~\cite{Mil89}, and have
proven to be a valuable tool in verifying the correctness of programs.
Approximate SR~\cite{GP07,PGT08,Tab09} have
been introduced in the control community as a generalization of SR in
order to enlarge the class of systems which admit discrete abstractions
(or symbolic models).
We refine the notion of aSR to acSR, with the aim of capturing a
contraction property that is often observed in concrete systems, see
e.g.~\cite{LS98,PvdWN05,PGT08}. Intuitively, the existence of a SR from system $S$ to system $\hat S$
implies that for every behavior of $S$ there exists a behavior of
$\hat S$ satisfying certain properties. In the classical setting, one
would ask that the output of the two related behaviors coincides, from
which behavioral inclusion follows.
For our purposes, as we want to preserve the
IODS inequality, we require that the input costs and output costs
satisfy $\hat I\le I$ and 
$O\le \hat O$ along those related behaviors. The satisfaction of these
inequalities allows us to conclude that $(\hat S,\hat I,\hat O)$ being pIODS implies that
$(S,I,O)$ is pIODS.

For notational convenience we use $R_X:=\pi_{X\times
\hat X}(R)$ to denote the projection of a relation \mbox{$R\subseteq X\times
\hat X\times U\times \hat U$} on \mbox{$X\times \hat X$}. Moreover, we use
$U(x)\defeq\{u\in U\mid
r(x,u)\neq\emptyset\}$ to denote the set of inputs  for which the
right-hand-side is non-empty.

\subsection{Exact simulation relations}

\begin{definition}\label{d:SR}
Let $S$ and $\hat S$ be two systems.
A relation $R\subseteq X\times \hat X\times U\times \hat U$ is said to
be a \emph{simulation
relation (SR) from $S$ to $\hat S$}  if: 
\begin{enumerate}

  \item for all $x_{0}\in X_{0}$ exists $\hat x_{0}\in \hat X_{0}$ such
  that $(x_{0},\hat x_{0})\in R_X$;
  \item for all $(x,\hat x)\in R_X$ and $u\in U(x)$
        there exists $\hat u\in \hat U(\hat x)$ such that 
  \begin{enumerate}
    \item $(x,\hat x,u,\hat u)\in R$;
    \item for all $x'\in r(x,u)$ there exists $\hat x'\in \hat r(\hat x,\hat u)$
    such that $(x',\hat x')\in R_X$.
  \end{enumerate}
\end{enumerate}

Let $(S,I,O)$ and $(\hat S,\hat I,\hat O)$ be two systems with cost
functions.  We call a SR $R$ form $S$ to $\hat S$ an
\emph{input-output SR (IOSR) from $(S,I,O)$ to $(\hat S,\hat I,\hat O)$} if 
  \begin{IEEEeqnarray}{rCl't'rCl}\label{e:costs}
    \hat I(\hat x, \hat u) &\le& I(x,u)
    & and&
    O(x,u) &\le& \hat O(\hat x,\hat u)
  \end{IEEEeqnarray}
holds for all $(x,\hat x,u,\hat u)\in R$.
\end{definition}

Note that the notion of IOSR for systems with input and
output costs 
is a straightforward extension of the well-known definition of SR for
the usual definition of system, see~\cite{Tab09}. 

\begin{lemma}\label{l:behaviors}
Let $S$ and $\hat S$ be two systems.
Suppose there exists an SR $R$ from $S$ to
$\hat S$, then for every behavior $(\xi,\nu)$
of $S$  there exists a  behavior $(\hat \xi,\hat \nu)$ of $\hat S$ such that 
\begin{IEEEeqnarray}{c'c}\label{e:behaviors}
  (\xi_{t},\hat \xi_{t},\nu_{t},\hat \nu_{t})\in R,&t\in\N.
\end{IEEEeqnarray}
\end{lemma}
\begin{proof}
The proof follows by similar arguments as the proof of
\cite[Proposition~4.9]{Tab09} and is omitted here.
\end{proof}

Simulation relations preserve IODS in the following sense.

\begin{theorem}\label{t:sim}
Let $(S,I,O)$ and $(\hat S,\hat I,\hat O)$ be two systems with cost
functions and
suppose there exists an IOSR $R$ from $(S,I,O)$ to
$(\hat S,\hat I,\hat O)$. If $(\hat S,\hat I,\hat O)$ is
$(\gamma,\mu,\rho)$-pIODS then $(S,I,O)$ is
$(\gamma,\mu,\rho)$-pIODS.
\end{theorem}
\begin{proof}
Since we assume that $(\hat S,\hat I,\hat O)$ is $(\gamma,\mu,\rho)$-pIODS, any behavior
$(\hat \xi,\hat \nu)$ of $\hat S$ satisfies
\begin{IEEEeqnarray*}{C}
\hat O(\hat \xi_{t},\hat \nu_t)\le \max_{t'\in \intcc{0;t}}
\mu(\gamma(\hat I(\xi_{t'},\hat \nu_{t'})),t-t')+\rho
\end{IEEEeqnarray*}
for all times $t\in\N$.
From Lemma~\ref{l:behaviors} follows that for every behavior
$(\xi,\nu)$ of $S$ there exists a behavior $(\hat \xi,\hat \nu)$ of
$\hat S$
 such that \eqref{e:behaviors} holds.
Now we derive the inequality
\begin{IEEEeqnarray*}{C}
\begin{IEEEeqnarraybox}[][c]{rCl}
O(\xi_{t},\nu_t)
&\le& 
\hat O(\hat \xi_{t},\hat \nu_t)\\
&\le& 
\max_{t'\in \intcc{0;t}} \mu(\gamma(\hat I(\hat\xi_{t'},\hat
\nu_{t'})),t-t')+\rho\\
&\le& 
\max_{t'\in \intcc{0;t}} \mu(\gamma(I(\xi_{t'}, \nu_{t'})),t-t')+\rho
\end{IEEEeqnarraybox}
\end{IEEEeqnarray*}
for all $t\in\N$. The last inequality follows from $\hat I(\hat
\xi_t,\hat \nu_t)\le I(\xi_t,\nu_t)$ and the monotonicity
properties of the functions $\gamma$ and $\mu$. Since we can repeat this argument for any 
behavior $(\xi,\nu)$ of $S$ we see that $(S,I,O)$ is
$(\gamma,\mu,\rho)$-pIODS.
\end{proof}

Note how preservation of pIODS is contra-variant, i.e., while the
direction of the simulation relation is from system $S$ to system
$\hat S$, the propagation of pIODS is from system $\hat S$ to
system~$S$. Moreover, by taking $\rho=0$ it follows that $\hat S$
being IODS implies $S$ is IODS.

\subsection{Approximate simulation relations}

Exact simulation relations are often too restrictive when one seeks to
relate a physical system to a finite-state abstraction or symbolic
model. In this case, approximate simulation relations were shown to be
adequate in the sense that they can be shown to exist for large
classes of physical systems~\cite{GP07,Tab09}.

\begin{definition}
Let $(S,I,O)$ and $(\hat S,\hat I,\hat O)$ be two systems with cost
functions.
A SR $R$ from $S$ to
$\hat S$ is called an \emph{$\varepsilon$-approximate input-output SR
($\varepsilon$-aIOSR) from $(S,I,O)$ to
$(\hat S,\hat I,\hat O)$} if every $(x,\hat x,u,\hat u)\in R$ satisfies:
\begin{IEEEeqnarray}{C}\label{e:acosts}
\begin{IEEEeqnarraybox}[][c]{rCl't'rCl}
  \hat I(\hat x,\hat u)
  &\le& I(x,u)+\varepsilon& and&
  O(x,u) &\le& \hat O(\hat x,\hat u)+\varepsilon.
\end{IEEEeqnarraybox}
\end{IEEEeqnarray}
\end{definition}

Note that the definition of aIOSR is again a straightforward extension
of the well-known notion of approximate SR of systems,
see~\cite{Tab09}. For $\varepsilon=0$ the notion of exact
IOSR is recovered. However, the notion of aIOSR introduces some
flexibility as it allows, for example, the inequality
$O(x,u)-\varepsilon\le \hat O(\hat x,\hat u) \le O(x,u) $ to hold which is not possible for IOSR.
This flexibility is important when we are dealing we infinite state systems
where an abstract state  in $\hat X$ corresponds to a set of
states in $X$.

\begin{theorem}
\label{t:asim}
Let $(S,I,O)$ and $(\hat S,\hat I,\hat O)$ be two systems with cost
functions and
suppose there exists an $\varepsilon$-aIOSR $R$ from $(S,I,O)$ and $(\hat S,\hat I,\hat O)$.
If $(\hat S,\hat I,\hat O)$ is $(\gamma,\mu,\rho)$-pIODS, then 
$(S,I,O)$ is $(\gamma',\mu,\rho')$-pIODS with $\gamma'(c)=2\gamma(2c)$ and
$\rho'=\mu(\gamma'(\varepsilon),0)+\varepsilon+\rho$.
\end{theorem}

\begin{proof}[Proof of Theorem~\ref{t:asim}]
Using the same arguments as in the proof of
Theorem~\ref{t:sim} we choose for any behavior
$(\xi,\nu)$ of $S$ a behavior $(\hat \xi,\hat \nu)$ of
$\hat S$ satisfying
\eqref{e:behaviors}. Then we obtain:
\begin{IEEEeqnarray}{c}\label{e:helper2}
\begin{IEEEeqnarraybox}[][c]{rCl}
  O(\xi_{t},\nu_t)
  &\le& 
  \hat O(\hat \xi_{t},\hat\nu_t)+\varepsilon\\
  &\le& 
  \max_{t'\in \intcc{0;t}} \mu(\gamma(\hat
  I(\xi_{t'},\nu_{t'})),t-t')+\varepsilon+\rho\\
  &\le& 
  \max_{t'\in \intcc{0;t}}
  \mu(\gamma(I(\xi_{t'},\nu_{t'})+\varepsilon),t-t')+\varepsilon+\rho.
\end{IEEEeqnarraybox}
\end{IEEEeqnarray}
Now we can use Lemma~\ref{l:kldbound} in the appendix with $\mu$, $\gamma$, and
$\varepsilon$ to obtain $\gamma'(r)= 2\gamma(2r)$, $\sigma(\varepsilon)=\mu(2\gamma(2\varepsilon),0)$ and
conclude
\begin{IEEEeqnarray*}{C}
\begin{IEEEeqnarraybox}[][c]{rCl,r*}
  O(\xi_{t},\nu_t) 
  &\le& 
  \max_{t'\in \intcc{0;t}} \mu(\gamma'(I(\xi_{t'},\nu_{t'})),t-t')
  +\sigma(\varepsilon)+\varepsilon+\rho
\end{IEEEeqnarraybox}
\end{IEEEeqnarray*}
which completes the proof.
\end{proof}

\subsection{Contractive simulation relations}

The construction of abstractions or symbolic models for physical
systems described in \cite{PGT08,PT09,Tab09} results in simulation
relations that satisfy a certain contraction property. Here we
introduce a notion of simulation that captures those
contraction properties.

In the following definition of contractive simulation relation from
$S$ to $\hat S$, we use a function $\sd{}:U\times \hat U\to \R_{\ge0}$ to measure the
``mismatch'' between two inputs $u\in U$ and $\hat u\in \hat U$. 
In various examples, in which we show that two systems are
related, the set of inputs $\hat U$ of system $\hat S$ is actually
a subset $\hat U\subseteq U$ of the set of inputs of system $S$ and we simply use a
norm $|\cdot|$ in $U$ as distance function $\sd{}(u,\hat u)=|u-\hat
u|$, see Example~\ref{ex:1}, Example~\ref{ex:2} and
Section~\ref{s:app}. However, in the following definition, we simply
assume we are given a function $\sd{}:U\times \hat
U\to \R_{\ge0}$ without referring to any underlying metric or norm.

\begin{definition}\label{d:acSR}
Let $S$ and $\hat S$ be two systems, 
let $\kappa,\lambda\in\R_{\ge0}$, $\beta\in\intco{0,1}$ be some
parameters and consider a map $\sd{}:U\times \hat U\to\R_{\ge0}$.
We call
a parameterized (by~$\varepsilon\in\intco{\kappa,\infty}$) relation
$R(\varepsilon)\subseteq X\times \hat X\times U\times \hat U$ a \emph{$\kappa$-approximate
$(\beta$,$\lambda)$-contractive simulation relation
($(\kappa,\beta,\lambda)$-acSR) from $S$ to $\hat S$ with distance
function $\sd{}$}  if
$R(\varepsilon)\subseteq R(\varepsilon')$ holds for all
$\varepsilon\le \varepsilon'$ and for all
$\varepsilon\in\intco{\kappa,\infty}$ we have
\begin{enumerate}
  \item $\forall x_{0}\in X_{0},\exists \hat x_{0}\in \hat
  X_{0}:(x_{0},\hat x_{0})\in R_X(\kappa)$;
  \item $\forall (x,\hat x)\in R_X(\varepsilon),\forall u\in
  U(x)$, $\exists \hat u\in \hat
  U(\hat x):$
\begin{enumerate}
\item $(x,\hat x,u,\hat u)\in R(\varepsilon)$
\item $\forall x'\in r(x,u),\exists \hat x'\in \hat r(\hat x,\hat
u):$\\[1ex]
     $(x',\hat x')\in R_X(\kappa+\beta \varepsilon+\lambda\sd{}(u,\hat u))$.
\end{enumerate}
\end{enumerate}

Let $(S,I,O)$ and $(\hat S,\hat I,\hat O)$ be two systems with cost
functions.
We call a $(\kappa,\beta,\lambda)$-acSR $R(\varepsilon)$ from $S$ to
$\hat S$ with distance function $\sd{}$ a \emph{$\kappa$-approximate $(\beta,\lambda)$-contractive input-output
SR ($(\kappa,\beta,\lambda)$-acIOSR) from $(S,I,O)$ to $(\hat S,\hat
I,\hat O)$ with
distance function $\sd{}$} if
there exist $\gamma_O,\gamma_I\in{\mathcal K}$ such that
\begin{IEEEeqnarray}{c}\label{e:accosts}
\begin{IEEEeqnarraybox}[][c]{rCl}
    \hat I(\hat x, \hat u) 
    &\le& 
    I(x,u)+
    \gamma_I(\varepsilon')\\
    O(x,u) 
    &\le&
    \hat O(\hat x,\hat u)+
    \gamma_O(\varepsilon')
\end{IEEEeqnarraybox}
\end{IEEEeqnarray}
holds for all $(x,\hat x,u,\hat u)\in R(\varepsilon)$ and
$\varepsilon'=\max\{\varepsilon,\sd{}(u,\hat u)\}$.
\end{definition}

Recall that in generalizing IOSR to aIOSR we merely relaxed the
inequalities on the costs functions by a constant parameter
$\varepsilon$, compare~\eqref{e:costs}
and~\eqref{e:acosts}. Here, we even go one step further, and relax the
inequalities using the generalized gain
functions $\gamma_I$ and $\gamma_O$, where $\varepsilon'$
in~\eqref{e:acosts} depends on the parameter $\varepsilon$ that
appears in the definition of the acSR $R(\varepsilon)$ and on the input
mismatch measured in terms of $\sd{}$. This change, in combination
with the definition of acSR, allows us to quantify the
relaxation in the cost function inequalities as a function of the
difference of input histories, see~Theorem~\ref{t:acbehaviors} and the
subsequent discussion. Before, we make those statements more precise,
let us first introduce an example to illustrate the notion of acSR.

\begin{example}\label{ex:1}
We consider a scalar disturbed linear system 
\begin{IEEEeqnarray}{c'c}\label{e:ex1:sys}
x^+=0.6 x+u.
\end{IEEEeqnarray}
on the bounded set $D:=\intcc{-1,1}$. 
We start our analysis by casting~\eqref{e:ex1:sys} as a system $S$
with $X:=\R$,
$X_0:=D$, $U:=\R$ and
$r(x,u):=\{0.6x+u\}$. 

Note that $D$ is forward invariant with respect
to~\eqref{e:ex1:sys} in the absence of disturbances, i.e., when $u=0$. Later on, we
analyze the invariance property in the presence of disturbances.
This motivates our choice of cost functions with $O(x,u):=|x|_D$ and
$I(x,u):=|u|$.

We now introduce a symbolic model $\hat S$ of $S$ with $\hat
X:=[D]_{0.2}$, $\hat X_0:=\hat X$, $\hat U:=\{0\}$ and
\begin{IEEEeqnarray*}{c}
  \hat x'\in \hat r(\hat x,\hat u):\iff |\hat x'-0.6\hat x|\le 0.2.
\end{IEEEeqnarray*}
Note that since $O(\hat x,\hat u)=I(\hat x,\hat u)=0$ for all $\hat x\in
\hat X$ and $\hat u\in\hat U$, we define the cost functions for $\hat
S$ to be $\hat O:=0$ and $\hat I:=0$. We also introduce the relation
$R(\varepsilon):=R_X(\varepsilon)\times\R\times\{0\}$ with
\begin{IEEEeqnarray*}{c}
  R_X(\varepsilon):=\{(x,\hat x)\in X\times\hat X\mid |x-\hat x|\le \varepsilon\}
\end{IEEEeqnarray*}
and show that $R(\varepsilon)$ is a
$(0.2,0.6,1)$-acSR from $S$ to $\hat S$ with distance function
$\sd{}(u,0):=|u|$. 

Point 1) in Definition~\ref{d:acSR} is easily
verified. Now let \mbox{$(x,\hat x)\in R_X(\varepsilon)$} and
$u\in U$. We pick $0\in \hat U$ and observe that $(x,\hat x,u,0)\in
R(\varepsilon)$ holds by definition of $R(\varepsilon)$. We proceed
with 2.b) of Definition~\ref{d:acSR}. For $x'\in r(x,u)$ there exists
$\hat x'\in \hat r(\hat x,0)$ with 
\begin{IEEEeqnarray*}{c}
|x'-\hat x'|\le 0.2+ |0.6 x+u-0.6\hat x|\le0.2+ 0.6 \varepsilon+|u|
\end{IEEEeqnarray*}
and it follows that $R(\varepsilon)$ is a
$(0.2,0.6,1)$-acSR from $S$ to $\hat S$.
Moreover, the inequalities~\eqref{e:accosts} are satisfied with
$\gamma_I=0$ and $\gamma_O(c)=c$.
Hence, $R(\varepsilon)$ is an acIOSR from $(S,I,O)$ to $(\hat S,\hat
I,\hat O)$.

Let us now emphasize that there exists no  $\varepsilon$-aIOSR $\hat R$ from
$(S,I,O)$ to $(\hat S,\hat I,\hat O)$ for any finite symbolic model $\hat S$.
For the sake of contradiction, suppose there exists  an
$\varepsilon$-aIOSR $\hat R$ from $(S,I,O)$ to $(\hat S,\hat I,\hat O)$ and $\hat S$ is
finite. Since $\hat S$ is
finite, there necessarily exists a state $\hat x\in \hat X$ and input $\hat
u\in \hat U$ such that
the set of related sates and inputs $\{(x,u)\in X\times U\mid (x,\hat
x,u,\hat u)\in \hat R\}$ is
unbounded. As a consequence, we find for any constant  $c\in\R$, a pair $(x,u)$ with $(x,\hat x,u,\hat u)\in \hat R$
so that $O(x,u)=|x|_D> \hat O(\hat x,\hat u)+c$ and $\hat R$ cannot be
an aIOSR since \eqref{e:acosts} is violated.

Conversely, if we bound the set of states and inputs of~\eqref{e:ex1:sys}
but consider the modified dynamics $x^+=x+u$, then it is easy to compute 
a relation $\hat R$ that is an $\varepsilon$-aIOSR from $(S,I,O)$ to $(\hat S,\hat I,\hat O)$, but there is no acSR from $\hat S$ to~$S$.

We resume the analysis of this example at the end of this
section, where we continue the robustness analysis of the invariance
property of $D$ with respect to $S$.
\end{example}

The previous example demonstrates that we can use acIOSR to relate an
infinite system $S$ with an \emph{unbounded} set of states and/or inputs,
with a \emph{finite} system $\hat S$, which is not possible using aIOSR.

We point out that any $(\kappa,\beta,\lambda)$-acIOSR $R(\varepsilon)$ from $(S,I,O)$ to $(\hat S,\hat I,\hat O)$ is
also an aIOSR, whenever the maximal distance between two related elements in $U$ and $\hat U$ is bounded. Let
$\alpha\in\R_{\ge0}$ be given such that $\sd{}(u,\hat u)\le \alpha$
holds for all $(u,\hat u)\in \pi_{U\times \hat U}(R(\varepsilon))$ and
$\varepsilon\in\R_{\ge0}$. Now we fix
$\varepsilon$ such that
$\kappa+\beta\varepsilon+\lambda\alpha\le\varepsilon$ holds. Note that
we can always find such an $\varepsilon$ as we assume
$\beta\in\intco{0,1}$.
Then the relation $R':=R(\varepsilon)$ is an aSR from $S$ to~$\hat S$. This observation follows
immediately from the definition of $R(\varepsilon)$ since 
$\kappa+\beta\varepsilon+\lambda\alpha\le\varepsilon$
implies that $R(\kappa+\beta\varepsilon+\lambda\alpha)\subseteq
R'$ which in turn implies that $R'$ is a SR
from $S$ to $\hat S$.
Moreover, if $R(\varepsilon)$ is an acIOSR then $R'$ is an
$\varepsilon'$-aIOSR from $(S,I,O)$ to $(\hat S,\hat I,\hat O)$ with
$\varepsilon':=\max\{\varepsilon,\gamma_O(\max\{\varepsilon,\alpha\}),$ $\gamma_I(\max\{\varepsilon,\alpha\})\}$.

Before we explain how the notions of acSR and acIOSR capture the
contraction property of $S$, we
provide a result that mimics Lemma~\ref{l:behaviors}.

\begin{theorem}\label{t:acbehaviors}
Let $S$ and $\hat S$ be systems and let
$R(\varepsilon)$ be a $(\kappa,\beta,\lambda)$-acSR from
$S$ to $\hat S$ with distance function $\sd{}$.
Then there exist 
$\mu_\Delta\in{\mathcal{KLD}}$ and $\gamma_\Delta,\kappa_\Delta \in\R_{\ge 0}$ 
such that for every  behavior $(\xi,\nu)$ of $S$ there
exists a behavior  $(\hat \xi,\hat \nu)$ of $\hat S$ so that the two
behaviors satisfy 
\begin{IEEEeqnarray}{c'c}\label{e:acbehaviors}
  (\xi_{t},\hat\xi_{t},\nu_{t},\hat\nu_{t})\in R(\varepsilon_t),&t\in\N.
\end{IEEEeqnarray}
with $\varepsilon_{t+1}\le\max_{t'\in\intcc{0;t}}\mu_\Delta(\gamma_\Delta \sd{}(\nu_{t'},\hat
\nu_{t'}),t-t')+\kappa_\Delta$.
\end{theorem}
\begin{proof}
First, we show by construction that for every behavior $(\xi,\nu)$ of
$S$ there exists a behavior $(\hat\xi,\hat\nu)$ of $\hat S$ such that $(\xi_{t},\hat\xi_{t},\nu_{t},\hat\nu_{t})\in R(\varepsilon_t)$ 
holds for all $t\in\N$ where $\varepsilon_t$ satisfies
\begin{IEEEeqnarray}{c'c}\label{e:scsys}
  \varepsilon_{t+1}=\kappa+\beta\varepsilon_t+\lambda\sd{}(\nu_t,\hat\nu_t),&
  \varepsilon_0=\kappa.
\end{IEEEeqnarray}
We define the sequences $\hat \xi:\N\to \hat X$ and
$\hat\nu:\N\to \hat U$ inductively. For the base case $t=0$, we choose
$\hat \xi_0\in
\hat X_0$ such that $(\xi_0,\hat \xi_0)\in R_X(\kappa)$ and
$\hat\nu_0\in \hat U$ such that $(\xi_0,\hat \xi_0,\nu_0,\hat\nu_0)$
satisfies 2.a) with $\varepsilon_0=\kappa$ and 2.b) of Definition~\ref{d:acSR}. Now suppose
$(\xi_{t'},\hat\xi_{t'},\nu_{t'},\hat\nu_{t'})$ satisfies 2.a) with
$\varepsilon_t$ satisfying~\eqref{e:scsys} and 2.b) of
Definition~\ref{d:acSR} for all $t'\in\intcc{0;t}$.
We choose $\hat\xi_{t+1}\in \hat r(\hat
\xi_t,\hat\nu_t)$ such that $( \xi_{t+1},\hat \xi_{t+1})\in
R_X(\varepsilon_{t+1})$ which in turn implies that we can fix $\hat
\nu_{t+1}\in \hat U$ such that $( \xi_{t+1},\hat
\xi_{t+1},\nu_{t+1},\hat \nu_{t+1})$ satisfies 2.a) with
$\varepsilon_{t+1}$ that satisfies \eqref{e:scsys} and 2.b) of
Definition~\ref{d:acSR}. It follows that $(\hat \xi,\hat \nu)$ is a
behavior of $\hat S$ and satisfies the claim.

In the remainder of the proof we use a discrete-time version of
\cite[Lemma~15]{Gru02} to show that here exist 
$\mu_\Delta\in{\mathcal{KLD}}$ and $\gamma_\Delta,\kappa_\Delta \in\R_{\ge 0}$ 
such that we have
\begin{IEEEeqnarray}{c}\label{e:h:acb}
\varepsilon_t
\le
\max_{t'\in\intcc{0;t}}\mu_\Delta(\gamma_\Delta \sd{}(\nu_{t'},\hat \nu_{t'}),t-t')+\kappa_\Delta.
\end{IEEEeqnarray}
We fix $\kappa_\Delta:=\kappa/(1-\beta)$,
$\gamma_\Delta:=\lambda/(\beta'-\beta)$ and
$g(c):=\beta'c$ for some $\beta'\in\intoo{\beta,1}$. Now it suffices to
verify that $V(\varepsilon):=|\varepsilon|_{\B(\kappa_\Delta)}$ satisfies 
$\gamma_\Delta|\sd{}(\nu_t,\hat \nu_t)|\le V(\varepsilon_{t})\implies
V(\varepsilon_{t+1})\le g(V(\varepsilon_t))$ holds.
Then it follows from~\cite[Lemma~15]{Gru02} that~\eqref{e:h:acb}
holds.
\end{proof}

Theorem~\ref{t:acbehaviors} exposes one of the key
features of an acIOSR. The membership \mbox{$(\xi_t,\hat \xi_t, \nu_t,\hat \nu_t)\in
R(\varepsilon_t)$} implies $O(\xi_t,\nu_t)\le \hat O(\hat
\xi_t,\hat \nu_t)+\gamma_O(\varepsilon_t)$. Hence, the  bound on the
output cost $O$ of $S$ in terms of the output cost $\hat O$ of $\hat
S$ depends on the
parameter $\varepsilon_t$ which is time-varying. In
comparison to the definition of aIOSR (see~\eqref{e:acosts}) this
parameter varies over time. We established with
Theorem~\ref{t:acbehaviors} a bound on $\varepsilon_t$
in terms of the difference (measured by $\lambda\sd{}$) of the input
histories $\sd{}(\nu_{t'},\hat \nu_{t'})$ with $t'\in\intcc{0;t}$. If we are
able to match a disturbance $\nu_t$ of $S$ closely (in terms of
$\sd{}$) by a disturbance $\hat \nu_t$ of $\hat S$, we know
that the output cost $\hat O$ of $\hat S$ provides a good estimate for the output
cost $O$ of $S$. Moreover, if after a certain $t'\in\N$ the
difference in the input behaviors is zero, i.e., $\sd{}(\nu_t,\hat
\nu_t)=0$ for all $t\ge t'$, then the bound on $\varepsilon_t$
approaches $\kappa_\Delta$ as $t\to\infty$. Here, we clearly exploit
the contraction parameter $\beta\in\intco{0,1}$ together with the
requirement 2.b) in the Definition~\ref{d:acSR} where the successor
states satisfy $(\xi_{t+1},\hat \xi_{t+1})\in
R(\kappa+\beta\varepsilon)$ whenever $(\xi_t,\hat \xi_t,\nu_t,\hat
\nu_t)\in R(\varepsilon)$ and $\sd{}(\nu_t,\hat\nu_t)=0$.

With the following corollary, we provide a bound on
$\varepsilon_t$ that depends solely on the behavior $(\xi,\nu)$ of
$S$ and not on the choice of a related behavior $(\hat \xi,\hat \nu)$
of $\hat S$.

\begin{corollary}\label{c:gains}
Given the premises of Theorem~\ref{t:acbehaviors}, let the function
$\Gamma:X\times
U\to\R_{\ge0}\cup\{\infty\}$ be given~by
\begin{IEEEeqnarray}{c}\label{e:deltaIC}
\Gamma(x,u):=
\sup\{\sd{}(u,\hat u)\mid \exists \varepsilon,\exists \hat x: (x,\hat x,u,\hat u)\in R(\varepsilon)\}.
\IEEEeqnarraynumspace
\end{IEEEeqnarray}
For any two behaviors $(\xi,\nu)$ and $(\hat \xi,\hat \nu)$ of $S$ and 
$\hat S$, respectively, that satisfy~\eqref{e:acbehaviors},
$\varepsilon_t$ in~\eqref{e:acbehaviors} is bounded by
\begin{IEEEeqnarray*}{c}
\varepsilon_{t+1}
\le
\max_{t'\in\intcc{0;t}}\mu_\Delta(\gamma_\Delta \Gamma(\xi_t,\nu_t),t-t')+\kappa_\Delta
\end{IEEEeqnarray*}
with $\kappa_\Delta=\kappa/(1-\beta)$,
$\gamma_\Delta=\lambda/(\beta'-\beta)$ and
$\mu_\Delta(r,t)=(\beta')^tr$ for any $\beta'\in\intoo{\beta,1}$.
\end{corollary}

We are now ready to state the main result of this section where we
show that pIODS is preserved under acIOSR. As in the
in case of SR and aSR, the proof strategy is to establish a pIODS
inequality for $S$ in terms of the pIODS inequality given for $\hat S$. For
acIODS, the estimates of the cost functions $I$ and $O$ in terms of
the cost functions $\hat I$ and $\hat O$ depend on the
time varying parameter $\varepsilon_t$. That is
reflected in the following theorem, by a modification of the input
costs $I$ of $S$ to $I'=\max\{I,\Gamma\}$. Here, $\Gamma$ is the
function that we used in  Corollary~\ref{c:gains} to
established a bound on $\varepsilon_t$. It represents the mismatch of
the inputs $U$ and $\hat U$ measured in terms of $\sd{}$.
\begin{theorem}
\label{t:acsim}
Let $(S,I,O)$ and $(\hat S,\hat I,\hat O)$ be systems with costs
functions and
suppose there exists a $(\kappa,\beta,\lambda)$-acIOSR
$R(\varepsilon)$ from $(S,I,O)$ to $(\hat S,\hat I,\hat O)$ with distance function $\sd{}$.
Then, $(\hat S,\hat I,\hat O)$ being pIODS implies that
$(S,I',O)$ is pIODS, with $I'(x,u)\defeq \max\{I(x,u),\Gamma(x,u)\}$
and $\Gamma$
given by~\eqref{e:deltaIC}.
\end{theorem}

In the proof of Theorem~\ref{t:acsim}, we use two lemmas,
Lemma~\ref{l:transformation} and Lemma~\ref{l:maxbound}, which are
given in the appendix.

\begin{proof}[Proof of Theorem~\ref{t:acsim}]
Let $(\xi,\nu)$ and $(\hat \xi,\hat \nu)$ be a behavior of $S$ of
$\hat S$, respectively, that satisfy~\eqref{e:acbehaviors}.
Using the fact that $\hat S$ is 
$(\hat\gamma,\hat\mu,\hat\rho)$-pIODS, \eqref{e:accosts},
and
Lemma~\ref{l:kldbound} we obtain
\begin{IEEEeqnarray*}{rCl}
  O(\xi_{t},\nu_t) 
  &\le&
  \hat O(\hat\xi_{t},\hat\nu_t)+\gamma_O(\varepsilon_t)\\
  &\le&
  \max_{t'\in \intcc{0;t}} \hat \mu(\hat
  \gamma(I(\xi_{t'},\nu_{t'})+\gamma_I(\varepsilon_t)),t-t')+ \gamma_O(\varepsilon_t)+\hat \rho\\
  &\le&
  \max_{t'\in \intcc{0;t}} \hat \mu(\hat\gamma'(I(\xi_{t'},\nu_{t'})),t-t')+\gamma_\varepsilon(\varepsilon_t)+\hat\rho
\end{IEEEeqnarray*}
with $\hat\gamma'(c)=2\hat \gamma(2c)$ and
$\gamma_\varepsilon(c)=\hat\mu(\hat\gamma'(\gamma_I(c)),0)+\gamma_O(c)$.
We use the bound on $\varepsilon_t$ from Corollary~\ref{c:gains}
and obtain
\begin{IEEEeqnarray}{c}\label{e:p:t:acsim:1}
\begin{IEEEeqnarraybox}[][c]{rCl}
  O(\xi_{t},\nu_t)
  &\le&
  \max_{t'\in \intcc{0;t}} 
  \hat \mu(\hat\gamma'(I(\xi_{t'},\nu_{t'})),t-t')\\
  &&+\gamma'_\varepsilon(\max_{t'\in\intcc{0;t}}\mu_\Delta(\gamma'_\Delta
  \Gamma(\xi_{t'},\nu_{t'}),t-t'))+\gamma'_\varepsilon(\kappa_\Delta)+\hat \rho
\end{IEEEeqnarraybox}
\end{IEEEeqnarray}
for $\gamma'_\Delta:=\max\{\gamma_\Delta,1\}$ and
  $\gamma'_\varepsilon(c):=\gamma_\varepsilon(2c)$.
We use Lemma~\ref{l:transformation} to choose $\mu'_\Delta\in
{\mathcal{KLD}}$
such that
$\gamma'_\varepsilon(\mu_\Delta(\gamma'_\Delta c,t)) =
\mu'_\Delta(\gamma'_\varepsilon(\gamma'_\Delta c),t)$.
Now we use Lemma~\ref{l:maxbound} to choose $\mu\in {\mathcal{KLD}}$ 
such that
\begin{IEEEeqnarray*}{c}
\max_{t'\in\intcc{0;t}}\hat \mu(c,t')
+\max_{t'\in\intcc{0;t}}\mu'_\Delta(c,t')
\le
\max_{t'\in\intcc{0;t}}\mu(2c,t')
\end{IEEEeqnarray*}
holds.
Then, by defining
$\gamma(c)\defeq$ $2\max\{\hat \gamma'(c),\gamma'_\varepsilon(\gamma'_\Delta c)\}$
the rhs of  \eqref{e:p:t:acsim:1} is bounded by
\begin{IEEEeqnarray*}{c}
  O(\xi_{t},\nu_t)
  \le
  \max_{t'\in\intcc{0;t}} \mu( \gamma( \max\{I( \xi_{t'},
  \nu_{t'}),\Gamma(\xi_{t'},\nu_{t'})\}),t-t')+\rho.
\end{IEEEeqnarray*}
with 
$\rho:=\gamma'_\varepsilon(\kappa_\Delta)+\hat \rho$.
\end{proof}

If the inequality $\hat I\le I$ holds, we can provide an pIODS type
inequality for $S$ that can be easily described in terms of the parameters
of the pIODS inequality of $\hat S$.

\begin{corollary}\label{c:acsim}
Given the premises of Theorem~\ref{t:acsim}, suppose $\gamma_O$
satisfies $\gamma_O(r+r')\le \gamma_O(r)+\gamma_O(r')$ and
that $\hat I(\hat x,\hat u)\le I(x,u)$ holds for all $(x,\hat x, u,\hat u)\in
R(\varepsilon)$ and $(\hat S,\hat I,\hat O)$ is $(\hat \gamma,\hat
\mu,\hat \rho)$-pIODS,  then every behavior $(\xi,\nu)$
of $S$ satisfies
\begin{IEEEeqnarray}{c}\label{e:cIODS} 
\begin{IEEEeqnarraybox}[][c]{rCl}\label{e:cIODS} 
  O(\xi_t,\nu_t)
  &\le&
  \max_{t'\in \intcc{0;t}}
  \hat \mu(\hat
  \gamma(I(\xi_{t'},\nu_{t'})),t-t')\:+ \\
  &&\max_{t'\in \intcc{0;t}}
  \gamma_O\big(\mu_\Delta(\gamma'_\Delta\Gamma(\xi_{t'},
  \nu_{t'}),t-t'))+\gamma_O(\kappa_\Delta\big)+\hat\rho
\end{IEEEeqnarraybox}
\end{IEEEeqnarray}
with 
$\gamma'_\Delta(r)=\max\{r,\gamma_\Delta(r)\}$, $\mu_\Delta$ and
$\kappa_\Delta$ from Corollary~\ref{c:gains}.
\end{corollary}

Even though in Theorem~\ref{t:acsim}, contrary to the
results in Theorem~\ref{t:sim} and Theorem~\ref{t:asim}, we do not
state the parameters $(\mu,\gamma,\rho)$ of the pIODS inequality for
$S$ in dependency of the parameters $(\hat \mu,\hat \gamma,\hat
\rho)$, inequality~\eqref{e:cIODS} provides us with some insights.
The first term in the inequality~\eqref{e:cIODS} follows from the fact that we were
able to successfully verify pIODS for $\hat S$. The
second term in~\eqref{e:cIODS} accounts for the ``mismatch'' between
the inputs $U$ and $\hat U$ . The last two terms,
i.e., the constant offset $\gamma_O(\kappa_\Delta)+\hat\rho$, is a result of the lower bound on the
parameter $\varepsilon\ge \kappa$ and $\hat \rho$ from the pIODS
inequality of $\hat S$.

Let us conclude this section with an application of 
Theorem~\ref{t:acsim} to Example~\ref{ex:1}. 
\addtocounter{example}{-1}
\begin{example}[continued]
Recall that, every
behavior $(\hat \xi,\hat \nu)$ of $\hat S$ satisfies 
$\hat O(\hat \xi_t,\hat \nu_t)= 0$
for all $t\in\N$. Therefore $(\hat S,\hat I,\hat O)$ is $(\hat \gamma,\hat \mu)$-IODS
with $\hat \gamma=\hat \mu=0$. We obtain $\Gamma$ for this example by
$\Gamma(x,u)=|u|$ and the input cost $I'$ coincides with
$I=\max\{I,\Gamma\}=I'$. In addition, the inequality $\hat I\le I$
holds and we can apply Corollary~\ref{c:acsim} to obtain the pIODS
inequality for every behavior $(\xi,\nu)$ of $S$ as
\begin{IEEEeqnarray}{c}\label{e:ex1:iods}
|\xi_t|_D\le \max_{t'\in\intcc{0;t}}\mu_\Delta(\gamma_\Delta
|\nu_{t'}|,t-t')+\kappa_\Delta
\end{IEEEeqnarray}
with $\kappa_\Delta=0.2/0.4$,
$\gamma_\Delta=1/(\beta'-0.6)$ and
$\mu_\Delta(r,t)=(\beta')^tr$ for any $\beta'\in\intoo{0.6,1}$.

Let us shortly describe how this inequality shows the robustness of
the invariance of $D$ with respect to $S$ against the disturbances
$\nu$. First, let us ignore the constant $\kappa_\Delta$ on the
right-hand-side of~\eqref{e:ex1:iods}. Then, the distance between the state
$\xi_t$ and $D$  is proportional to the norm of the disturbance
$\nu_t$. Moreover, the effect of a disturbance at some time $t'$
disappears over time since $\beta^{t-t'}\gamma_\Delta|\nu_{t'}|$
approaches zeros as $t\to\infty$. The constant $\kappa_\Delta$ appears
in~\eqref{e:ex1:iods} because we established the inequality through
the use of the symbolic model $\hat S$ and represents the effect of
quantization.
\end{example}

\section{Controller Design} 
\label{s:asr}

So far we interpreted the set of inputs $U$ of a 
system $S$ as disturbance inputs over which we had no control. However, in
this section, we assume that the input set $U$ is composed of
a set of {\it control} inputs $U^c$ and a set of {\it disturbance}
inputs $U^d$, i.e., $U=U^c\times U^d$. Moreover, we introduce a
controller that is allowed to modify the system behavior by imposing 
\emph{restrictions} on the control inputs~$U^c$.
In our framework, a \emph{controller for $S$} consists
of a system $S_C$ and a relation $R_{C}$. The \emph{controlled 
system} $S_C\times_{R_C} S$ is given by the composition of
$S_C$ with $S$ where $R_{C}$ is used to restrict the control inputs $U^c$
depending on the current state of $S_C$ and $S$.

In~\cite{TCRM14}, a synthesis approach has been developed to construct
a controller $(\hat S_C,\hat R_{C})$ rendering a \emph{finite} system
$\hat S$ IODS, i.e., the composed system $\hat S_C\times_{\hat
R_{C}}\hat S$ is IODS\footnote{Technically, the controller
in~\cite{TCRM14} is
defined in a slightly different manner from $(\hat S_C,\hat R_{C})$. However, it is straightforward to obtain a controller $(\hat
S_C,\hat R_{C})$ from the controller given in~\cite{TCRM14}.}.
In order to apply those results to a (possibly infinite) CPS $S$ we
first compute a finite symbolic model  $\hat S$ of $S$ and then
provide a procedure to \emph{transfer} (or \emph{refine})  a
controller  $(\hat S_C, \hat R_{C})$ that is designed
for $\hat S$ to a controller $(S_C,R_{C})$ for $S$.
This brings us to the main question answered in this section:

\vspace{1ex}
  \emph{Given $(S,I,O)$, what are the conditions that a symbolic model
  $(\hat S,\hat I,\hat O)$ of $(S,I,O)$ needs to satisfy so that the existence of a controller $(\hat
  S_C,\hat R_{C})$ for
  $\hat S$ rendering  $\hat S_C\times_{\hat R_{C}}\hat S$ pIODS, implies the existence of a controller
  $(S_C,R_{C})$ for $S$ rendering $S_C\times_{R_{C}} S$ pIODS?} 
\vspace{1ex}

A well-known approach for
controller refinement in connection with symbolic models is based on 
\emph{alternating simulation relations (ASR)}, see~\cite{AHKM98} and \cite[Chapter~4.3]{Tab09}.
In this section, we extend this approach to approximate contractive
alternating input-output SR (acAIOSR). An intuitive version of the main result proved in this section is:

\vspace{1ex}
  \emph{Consider two systems $(S,I,O)$ and $(\hat S,\hat I,\hat O)$, and let $R$ be an
acAIOSR from $(\hat S,\hat I,\hat O)$ to $(S,I,O)$. Suppose there exists a 
controller $(\hat S_C,\hat R_{C})$ for $\hat S$ such that $(\hat S_C\times_{\hat R_{C}}\hat S,\hat I,\hat O)$ is pIODS. Then there exist a
controller $(S_C,R_{C})$ for $S$ such
that $(S_C\times_{R_{C}} S,I,O)$ is pIODS.} 
\vspace{1ex}

We provide a precise formulation of this statement in
Theorem~\ref{t:main}, after we  formalize the notions of
acAIOSR, controller, and composition of a system with a controller.
Moreover, we explain how $(S_C,R_{C})$ can be constructed from $(\hat
S_C, \hat R_{C})$.

\subsection{Alternating simulation relations}
\label{ss:ASR}

In the following definition of an ASR we use a refined notion of input sets associated to states given by: 
$$U^c(x)\defeq\{u^c\in U^c\mid \forall u^d\in U^d:
r(x,u^c,u^d)\neq\emptyset\}.$$ 

\begin{definition}\label{d:acASR}
Let $S$ and $\hat S$ be two systems,
let $\kappa,\lambda\in\R_{\ge0}$ and $\beta\in\intco{0,1}$ be some
parameters and consider the map $\sd{}:\hat U\times U\to\R_{\ge0}$.
We call a parameterized (by~$\varepsilon\in\intco{\kappa,\infty}$) relation
$R(\varepsilon)\subseteq \hat X\times X\times \hat U\times U$ a \emph{$\kappa$-approximate
$(\beta$,$\lambda)$-contractive alternating simulation relation
($(\kappa,\beta,\lambda)$-acASR) from $\hat S$ to $S$ with 
distance function $\sd{}$} if
$R(\varepsilon)\subseteq R(\varepsilon')$ holds for all
$\varepsilon\le \varepsilon'$ and we have for all~$\varepsilon\in\intco{\kappa,\infty}$ 
\begin{enumerate}
  \item $\forall \hat x_{0}\in \hat X_{0},\exists x_{0}\in
  X_{0}:(\hat x_{0}, x_{0})\in R_X(\kappa)$;
  \item $\forall (x,\hat x)\in R_X(\varepsilon),\forall \hat u^c\in
    \hat U^c(\hat x),\exists u^c\in U^c(x),$
  \begin{enumerate}
    \item $\forall u^d\in U^d,\exists \hat u^d\in \hat U^d:$
    \begin{itemize}
    \item $(\hat x,x,\hat u,u)\in R(\varepsilon)$;
    \item $\forall x'\in r(x,u),\exists \hat x'\in \hat r(\hat x,\hat
    u):$\\[1ex]
     $(\hat x',x')\in R_X(\kappa+\beta \varepsilon+\lambda\sd{}(\hat u,u))$;
    \end{itemize}
  \end{enumerate}
   with $u\defeq (u^c,u^d)$, $\hat u\defeq (\hat u^c, \hat u^d)$.
\end{enumerate}
Let $(S,I,O)$ and $(\hat S,\hat I,\hat O)$ be two systems with cost
functions.
We call a $(\kappa,\beta,\lambda)$-acASR $R(\varepsilon)$ from $\hat
S$ to $S$ with distance function $\sd{}$ a \emph{$\kappa$-approximate $(\beta,\lambda)$-contractive
alternating input-output
SR ($(\kappa,\beta,\lambda)$-acAIOSR) from $(\hat S,\hat I,\hat O)$ to $(S,I,O)$ with
 distance function $\sd{}$} if
there exist $\gamma_O,\gamma_I\in{\mathcal K}$ such that
\begin{IEEEeqnarray}{c}\label{e:costASR}
\begin{IEEEeqnarraybox}[][c]{rCl}
    \hat I(\hat x, \hat u) 
    &\le& 
    I(x,u)+
    \gamma_I(\varepsilon')\\
    O(x,u) 
    &\le&
    \hat O(\hat x,\hat u)+
    \gamma_O(\varepsilon')
\end{IEEEeqnarraybox}
\end{IEEEeqnarray}
with $\varepsilon':=\max\{\varepsilon,\sd{}(\hat u,u)\}$
holds for all $(\hat x,x,\hat u,u)\in R(\varepsilon)$.

We call a relation $R(\varepsilon)$ \emph{acASR (acAIOSR)} if there exists
 $\beta\in\intco{0,1}$, $\kappa,\lambda\in\R_{\ge0}$ such
that $R(\varepsilon)$ is a $(\kappa,\beta,\lambda)$-acASR (acAIOSR) from $\hat S$
to~$S$ ($(\hat S,\hat I,\hat O)$ to $(S,I,O)$).
\end{definition}

We illustrate acAIOSR using an example from the literature.

\begin{example}[DC-DC boost converter]\label{ex:2}
We consider a popular
example from the literature, the boost DC-DC converter, see for
example~\cite{GPT10,MABBPWBCGFJKMPRR10}. The dynamics of the boost converter is given by a
two-dimensional switched
linear system $\dot \xi(t)=\bar A_u \xi(t)+\bar B$ with $\bar A_u\in\R^{2\times 2}$,
$\bar B\in\R^2$ and $u\in\{1,2\}$. In
\cite{GPT10} a symbolic model $\hat S$ of the sampled dynamics of the boost
converter $S$ is used to compute a controller rendering the set $D=\intcc{1.3,\,1.7}\times\intcc{5.7,\,5.8}$
positively invariant. Similarly to the approach in this paper, a
symbolic model $\hat S$ together with an approximate ASR $\hat R$ is first computed.
In the second step, a controller $(\hat S_C,\hat R_{C})$ for $\hat S$ is
computed to render $D$ positively invariant with respect to the
symbolic model $\hat S_C\times_{\hat R_{C}} \hat S$. Afterwards, a
controller for $S$ is obtained by refining the controller $(\hat
S_C,\hat R_{C})$.

Note, as the controller refinement in \cite{GPT10} is based on an
$\varepsilon$-ap\-proxi\-mate ASR with constant $\varepsilon\in\R_{\ge0}$, a disturbance $w\in\R^2$ on the system dynamics 
$\dot \xi(t)=\bar A_u \xi(t)+\bar B+w$ might lead to a state $\xi(\tau)$ such
that the composed system is blocking. Therefore, the resulting controller is prone to fail in
the presence of disturbances. Contrary to that, we exploit the
contractivity of the matrices $\bar A_u$ and construct a robust
controller using the introduced notion of acAIOSR.

We refer the reader to \cite{MABBPWBCGFJKMPRR10} for a detailed exposition of
the boost converter. In this example, we simply use the same
parameters as in~\cite{GPT10}, and obtain the sampled dynamics of
the boost converter as
$\xi_{t+1}= A_{\nu_t}\xi_t+ B_{\nu_t}+\omega_t$
with the system matrices given by
\begin{IEEEeqnarray*}{rCl,rCl}
   A_1&=&\begin{bmatrix} 0.9917 & 0 \\ 0 & 0.9964 \end{bmatrix}, &
   B_1&=&\begin{bmatrix} 0.1660 \\ 0  \end{bmatrix},\\
   A_2&=&\begin{bmatrix} 0.9903 & -0.0330 \\ 0.0354 & \phantom{-}0.9959 \end{bmatrix}, &
   B_2&=&\begin{bmatrix} 0.1659 \\ 0.0030  \end{bmatrix}.
\end{IEEEeqnarray*}
Note that in contrast to \cite{GPT10} we add $\omega_t\in\R^2$ to model
various disturbances. We introduce the system
$S=(X,X_{0},U,r)$ associated with the boost converter by defining
$X:=\R^2$,
$X_{0}:=D$,
$U:=U^c\times U^d$ with $U^c:=\{1,2\}$ and $U^d:=\R^2$. Note that the
inputs $(u^c,u^d)\in U$ of the system $S$ correspond to the control
input $u^c=u$ and the disturbance $u^d=w$. The transition
function is given by $r(x,(u^c,u^d)):=\{ A_{u^c}x+ B+u^d\}$. We use the cost functions
$I(x,(u^c,u^d)):=|u^d|$ and $O(x,u):=|x|_D$ to quantify the
disturbances and to encode the desired behavior.

The symbolic model
$\hat S=(\hat X,\hat X_{0},\hat U,\hat r)$ that is used in
\cite{GPT10} is based on a discretization of $D$:
\begin{IEEEeqnarray*}{c}
\hat X:=\hat X_{0}:=
D\cap\{x\in\R^2\mid x_i=k_i2/\sqrt{2}\kappa,i\in\{1,2\},k_i\in\Z\}
\end{IEEEeqnarray*}
with $\kappa=0.25\cdot
10^{-3}/\sqrt2$.
The inputs are given
by $\hat U:=\hat U^c\times \hat U^d$ with $\hat U^c:=\{1,2\}$ and
$\hat U^d:=\{0\}$. The transition function is implicitly given by
$\hat x'\in\hat r(\hat x,(\hat u,0))\iff |\hat x'- A_{\hat u}\hat x- B_{\hat u}|_2\le \kappa$. 

We set the cost functions for $\hat S$ simply to $\hat I(\hat
x,\hat u):=0$ and $\hat O(\hat x,\hat u):=0$ since $I(\hat
x,\hat u)=O(\hat x,\hat u)=0$ holds for all $\hat x$ and~$\hat u$. Let us introduce the
relation $R(\varepsilon):=R_X(\varepsilon)\times R_U$ with
\begin{IEEEeqnarray*}{c't'c}
R_X(\varepsilon)
:=\{(\hat x, x)\in\hat X\times X\mid |\hat x- x|_2\le \varepsilon\}\\
R_U:=\{((\hat u^c,0),(u^c,u^d))\in \hat U\times U\mid u^c=\hat u^c\}.
\end{IEEEeqnarray*}
We now show that $R(\varepsilon)$
is a $(\kappa,\beta,\lambda)$-acAIOSR from $\hat S$ to~$S$ with $\sd{}((u^c,u^d),(\hat u^c,0)):=|u^d|_2$ for
$\beta=0.997\ge \max\{|A_1|_2,|A_2|_2\}$ and $\lambda=1$. We first note that
$R(\varepsilon)\subseteq R(\varepsilon')$ holds whenever
$\varepsilon\le \varepsilon'$. By definition of $\hat X_{0}$ we can
see that for every $\hat x_0\in \hat X_{0}$ there exists a $x_0\in X_{0}$ such
that $(\hat x_0,x_0)\in R_X(\kappa)$. We proceed by checking 2) of
Definition~\ref{d:acASR}. Let $(\hat x, x)\in R_X(\varepsilon)$ and
$\hat u^c\in\hat U^c$. We choose $u^c=\hat u^c$ and observe that for every
$u^d\in U^d$ we have $(\hat x,x,(\hat u^c,0),(u^c,u^d))\in R(\varepsilon)$ and
$(\hat x', x')\in R_X(\kappa+\beta\varepsilon+\lambda|u^d|_2)$ with
$x'\in r(x,(u^c,u^d))$, $\hat
x'\in \hat r(\hat x,(\hat u^c,0))$ since 
\begin{IEEEeqnarray*}{c}
|x'-\hat x'|_2
\le
\kappa+ |A_{u^c}x+u^d-A_{u^c}\hat x|_2
\le
\kappa+\beta\varepsilon+|u^d|_2
\end{IEEEeqnarray*}
which shows that $R(\varepsilon)$ is an
$(\kappa,\beta,\lambda)$-acASR from $\hat S$ to $S$. As the
inequalities~\eqref{e:costASR} hold for $\gamma_I=0$ and $\gamma_O(c)=c$
we conclude that $R(\varepsilon)$ is an acAIOSR from $(\hat S,\hat I,\hat O)$ to
$(S,I,O)$. 

Similarly to previous examples, we exploited the contraction property of
the control system to construct an acASR from the symbolic model $\hat
S$ to $S$.

We resume the example after we presented the main
theorem of this section, where we refine the controller for the
symbolic model $\hat S$ to a controller for $S$.
\end{example}

\subsection{System composition}

In this subsection, we define a general
notion of system composition between two systems $S_1$ and $S_2$ with respect to
a relation $H\subseteq X_1\times X_2\times U_1\times U_2$. Afterwards,
we introduce the notion of system composition for the case when $H$ is an
acASR $R(\varepsilon)$ from $S_1$ to $S_2$. In the next subsection, we use the definition of system composition to define the controlled
system.

\begin{definition}\label{d:comp}
The \emph{composition of system $S_1$ and $S_2$ with respect to
the relation $H\subseteq X_1\times X_2\times U_1\times
U_2$}, is denoted by $S_{12}:=S_1\times_H S_2$ and defined by:
\begin{enumerate}
  \item $X_{12}:=X_1\times X_2$;
  \item $X_{120}:=(X_{10}\times X_{20})\cap H_X$;
  \item $U_{12}:=U_1\times U_2$;
  \item $(x'_1,x'_2)\in r_{12}((x_1,x_2),(u_1,u_2)):\iff$
    \begin{enumerate}
      \item $x'_2\in r_1(x_1,u_1)$;
      \item $x'_1\in r_2(x_2,u_2)$;
      \item $(x_1,x_2,u_1,u_2)\in H$ and $(x'_1,x'_2)\in H_X$.
    \end{enumerate}
\end{enumerate}

If $H$ is an $(\kappa,\beta,\lambda)$-acASR $R(\varepsilon)$ from $S_1$ to
$S_2$ with distance function $\sd{}$, then we exchange 2) by
$X_{120}:=(X_{10}\times X_{20})\cap R_X(\kappa)$ and
4.c) by 
\begin{IEEEeqnarray*}{c't'c}
 (x_1,x_2,u_1,u_2)\in R(e(x_1,x_2)),& and &
      (x'_1,x'_2)\in
      R_X(\varepsilon')
\end{IEEEeqnarray*}
with $\varepsilon':=\kappa+e(x_1,x_2)\beta+\lambda\sd{}(u_1,u_2)$
and $e(x_1,x_2):=\inf\{\varepsilon\in\R_{\ge0}\mid (x_1, x_2)\in
R_X(\varepsilon)\}$.

\end{definition}

Intuitively, our definition of system composition corresponds to the
well-known definition of parallel composition
of the systems $S_1$ and $S_2$ with synchronization defined by $H$,
respectively $R(\varepsilon)$.
The only transitions allowed on the composed system
$S_1\times_{H} S_2$ are those for which the corresponding
states and inputs belong to $H$, i.e., \mbox{$(x_1,x_2, u_1, u_2)\in
H$}. It is shown in~\cite{Tab09} how this notion of composition can
describe series, parallel, feedback and several other
interconnections. For the case that $H$ is an acASR $R(\varepsilon)$, we
require that $(x_1,x_2, u_1, u_2)\in
R(\varepsilon)$ where we fix $\varepsilon=e(x_1,x_2)$.  With our particular
choice of $\varepsilon=e(x_1, x_2)$ we restrict the transitions of the
composed system $S_1\times_{R(\varepsilon)} S_2$ to those states and
inputs that are related by the smallest $\varepsilon=e(x_1,x_2)$
possible. 
In general it is not ensured that the infimal $\varepsilon=e(x_1,x_2)$
is actually attained by the states $(x_1,x_2)$. Therefore, we
assume in the following that 
\begin{IEEEeqnarray}{c}\label{e:infR}
  e(x_1,x_2)<\infty \implies (x_1, x_2)\in R_X(e(x_1, x_2)).
\end{IEEEeqnarray}
Note that this assumption is often satisfied in practice where
$R_X(\varepsilon)$ is for example defined by $|x_1-x_2|\le \varepsilon$.

\subsection{The controlled system and controller refinement}

In the following, we use the composition of two systems
$S_C$ and $S$ with respect to a parameterized relation
$R_C(\varepsilon)$ to define the controlled system
$S_C\times_{R_{C}(\varepsilon)} S$, when the relation
$R_C(\varepsilon)$ is an acASR from $S_C$ to $S$.
From a control perspective, the controller $(S_C,R_{C}(\varepsilon))$
for $S$ can be
implemented in a feedback loop as follows. Let us denote the set of initial states
$x\in X_0$ for which there exists $x_C\in X_{C0}$ such that  $(x_C, x)\in
R_{C,X}(\kappa)$ by $X'_{0}$.  
Then initially, {\bf i)} the controller measures the system state $x\in
X_{0}'$ and determines a related controller state
\mbox{$x_C\in  X_{C0}$} such that  $(x_C, x)\in R_X(\kappa)$;
{\bf ii)} the controller picks the control inputs $u_C^c$ and $u^c$ according
to 2) in Definition~\ref{d:acASR} and applies $u^c$ to $S$; {\bf iii)}
the disturbance chooses $u^d\in U^d$ and $x'\in
r(x,(u^c,u^d))$; {\bf
iv)} the controller measures the new state $x'$ and chooses $ x_C'$
and $u_C^d\in  U_C^d$ such that $ x_C'\in 
r_C( x_C,(u_C^c, u_C^d))$ and $(x_C', x')\in
R_X(\varepsilon')$ for $\varepsilon'=e(x_C', x')$. Now the cycle
continues with ${\bf ii)}$.

Note that in this scenario, the disturbance inputs $U_C^d$ of the
controller $S_C$ are not considered as external inputs, but are
allowed to be chosen by the controller. This leads us to the following
the definition.

\begin{definition}\label{d:controller}
Given a system $S$, we call the pair $(S_C,R_{C}(\varepsilon))$ a \emph{controller
for $S$} if $S_C$ is a system, $R_{C}(\varepsilon)$ is an acASR from
$S_C$ to $S$
and the composed system
$S_C\times_{R_{C}(\varepsilon)} S$ is
non-blocking, in the sense that for
all reachable states $(x_C, x)$ there exists $(u_C^c, u^c)\in
U_C^c\times U^c$ such that for all $u^d\in U^d$ there exists $u_C^d\in  U_C^d$ for which $r'((x_C,x),((u^c_C,u^d_C),(u^c,u^d)))\neq\emptyset$, where $r'$ is
the transition map of the composed system.
\end{definition}

The interested reader may wish to consult~\cite[Chapter~6.1]{Tab09} for detailed
explanations of 
why the composition between a controller and a system
is only well defined when the relation $R_C$ is alternating. Note that
the assumption~\eqref{e:tightASR} is consistent with the use of 
extended alternating simulation relations in the definition of the
feedback composition in~\cite[Definition~6.1]{Tab09}.

Let us remark that the controller $(\hat S_C,\hat R_C)$ rendering the
system $\hat S$ pIODS that we obtain from the approach
in~\cite{TCRM14} is given in terms of a system $\hat S_C$ and an
\emph{alternating simulation relation} (ASR) from $\hat S_C$ to $\hat
S$ rather than an acASR. The definition of an ASR is given in
\cite[Definition 4.22]{Tab09}. Instead of repeating the definition
here, we define it in terms of an acASR.

\begin{definition}
Let $S$ and $\hat S$ be two systems and let $R(\varepsilon)$ be
a $(0,0,0)$-acASR from $\hat S$ to $S$. The relation $\hat R:=R(0)$ is
called an \emph{al\-ter\-nating simulation relation (ASR)} from $\hat S$
to $S$.  
\end{definition}

The composition $S_1\times_{R_{12}}S_2$ of $S_2$
and $S_1$ with respect to an ASR $R_{12}$
follows from Definition~\ref{d:comp} with $H=R_{12}$. Similarly, the
definition of a controller $(S_C,R_{C})$ in terms of an ASR
follows in a straightforward manner from
Definition~\ref{d:controller}. No confusion between acASR and ASR should arise, since we
always include the parameter $\varepsilon$ in the notation when we refer to an
acASR (acAIOSR). 

In the following, we assume that an ASR
$R_{12}$ from
$S_2$ to $S_1$ satisfies
\begin{IEEEeqnarray}{c't}\label{e:tightASR}
(x_1,x_2,(u^c_1,u^d_1),(u^c_2,u^d_2))\in R_{12}\implies&
$(x_1,x_2,u^c_1,u^c_2)$ satisfies 2.a) of Def.~\ref{d:acASR}.
\end{IEEEeqnarray}
This implication~\eqref{e:tightASR} results in no loss of generality
since we can always construct an  ASR
$R'_{12}$ that satisfies~\eqref{e:tightASR} from an ASR
$R_{12}$ by simply removing the elements that don't satisfy~\eqref{e:tightASR}.

Given a system with cost functions $(S,I,O)$ and a
controller $(S_C,R_C)$ for $S$, we abuse the notation and use $(S_C\times_{R_{C}} S, I,O)$
to refer to the composed system $S_C\times_{R_{C}} S$ with cost functions
$I_{C}((x_C,x),(u_C,u)):=I(x,u)$ and 
$O_{C}((x_C,x),(u_C,u)):=O(x,u)$.

Like in Corollary~\ref{c:gains}, we define the function 
\begin{IEEEeqnarray}{c}\label{e:deltaI}
\Gamma(x,u):= \sup\{\sd{}(\hat u,u)\mid \exists
\varepsilon,\exists \hat x: (\hat x,x,\hat u,u)\in R(\varepsilon)\}
\end{IEEEeqnarray}
for an acAIOSR $R(\varepsilon)$ from $\hat S$ to $S$ with distance function $\sd{}$
and refer to 
$R(\varepsilon)$ as \emph{acAIOSR from $\hat S$ to
$S$ with~$\Gamma$}.

Now we are ready to state the main theorem.

\begin{theorem}\label{t:main}
Given two systems with cost functions $(S,I,O)$ and $(\hat S,\hat I,\hat O)$, let $R(\varepsilon)$ be an
acAIOSR from $(\hat S,\hat I,\hat O)$ to $(S,I,O)$ with $\Gamma$ and
let $R(\varepsilon)$ satisfy~\eqref{e:infR}. Suppose there exists a 
controller $(\hat S_C,\hat R_{C})$ for $\hat S$ with $\hat R_C$
satisfying~\eqref{e:tightASR} and such that $(\hat
S_C\times_{\hat R_{C}}\hat S,\hat O,\hat I)$ is pIODS. Then there exists a
controller $(S_C,R_{C}(\varepsilon))$ for $S$ such
that $(S_C\times_{R_{C}(\varepsilon)} S,I',O)$ is pIODS with 
\mbox{$I':=\max\{I,\Gamma\}$}.
\end{theorem}

We use the following lemmas, whose proofs are given in the
Appendix, to prove
Theorem~\ref{t:main}.

\begin{lemma}\label{l:transfer}
Consider the systems $S_1$, $S_2$ and $S_3$. 
Let $R_{12}$ be an ASR from $S_1$ to $S_2$ that
satisfies~\eqref{e:tightASR} and let $R_{23}(\varepsilon)$ be
a $(\kappa,\beta,\lambda)$-acASR from $S_2$ to $S_3$ with
$\sd{23}$ satisfying~\eqref{e:infR}.
Then there exists a $(\kappa,\beta,\lambda)$-acASR
$R_{123}(\varepsilon)$ from $S_{12}=S_1\times_{R_{12}} S_2$ to
$S_3$ with distance function
$\sd{{123}}((u_1,u_2),u_3):=\sd{{23}}(u_2,u_3)$
satisfying~\eqref{e:infR}.
\end{lemma}

\begin{lemma}\label{l:controllerNB}
Consider the systems $S_1$ and $S_2$. Let $R_{12}(\varepsilon)$ be an
acASR from $S_1$ to $S_2$ that satisfies~\eqref{e:infR}. Then
$(S_1,R_{12}(\varepsilon))$ is a controller for $S_2$.
\end{lemma}

\begin{lemma}
\label{l:closedloop}
Consider the systems $S_1$ and $S_2$. Let $R_{12}(\varepsilon)$ be a
$(\kappa,\beta,\lambda)$-acASR from $S_1$ to $S_2$ with
$\sd{12}$ satisfying~\eqref{e:infR}. Then there exists a
$(\kappa,\beta,\lambda)$-acSR $R_{121}(\varepsilon)$ from
\mbox{$S_{12}=S_1\times_{R_{12}(\varepsilon)}S_2$} to $S_1$ with
distance function $\sd{121}((u_1,u_2),u_1'):=\sd{{12}}(u_1, u_2)$.
\end{lemma}

\begin{proof}[Proof of Theorem~\ref{t:main}]
We apply Lemma~\ref{l:transfer} for $S_1=\hat S_C$, $S_2=\hat S$,
$S_3=S$, $R_{12}=\hat R_C$ and $R_{23}(\varepsilon)=R(\varepsilon)$.
It follows that there exists
an acASR $R_{C}(\varepsilon)$ from $\hat
S_C\times_{\hat R_C}\hat S$ to $S$ with distance function
$\sd{C}((\hat u_C, \hat u),u):=\sd{}(\hat u,u)$ and $R_C(\varepsilon)$
satisfies~\eqref{e:infR}.  
We apply Lemma~\ref{l:controllerNB} to see that $(S_{C},R_{C}(\varepsilon))$
with $S_C:=\hat
S_C\times_{\hat R_C}\hat S$ is a
controller for $S$. Now it follows from Lemma~\ref{l:closedloop}
that there exists an acSR $R'(\varepsilon)$ from
$S_C\times_{R_{C}(\varepsilon)}S$
to $S_C=\hat S_C\times_{\hat R_C}\hat S$
with distance function $\sd{}'(((\hat u_C,\hat
u),u),(\hat u'_C,\hat u'))=\sd{C}((\hat u_C,\hat
u),u)=\sd{}(\hat u,u)$.

Note that the cost functions for the composed systems
$\hat S_C\times_{\hat R_C} \hat S$ and $S_C\times_{R_C(\varepsilon)}
S$ are given~by
\begin{IEEEeqnarray*}{c'c}
\hat I_C((\hat x_C,\hat x),(\hat u_C,\hat u))=\hat I(\hat x,\hat u), &
\hat O_C((\hat x_C,\hat x),(\hat u_C,\hat u))=\hat O(\hat x,\hat u),\\
I_C((u_C,x),(u_C,u))=I(x,u),& O_C((x_C,x),(u_C,u))=O(x,u).
\end{IEEEeqnarray*}
We proceed by showing that $R'(\varepsilon)$ is actually a
$(\kappa,\beta,\lambda)$-acIOSR form $(S_C\times_{R_C(\varepsilon)}
S,I_C,O_C)$ to 
$(\hat S_C\times_{\hat R_C} \hat S,\hat I_C,\hat O_C)$.
By carefully checking the proof of the Lemmas~\ref{l:transfer}
and~\ref{l:closedloop}, we see that
\mbox{$((x_C,x),(\hat x_C,\hat x),(u_C,u),(\hat u_C,\hat u))\in R'(\varepsilon)$} implies
$x_C=(\hat x_C,\hat x)$, $u_C=(\hat u_C,\hat u)$ and $(\hat
x,x,\hat u,u)\in R(\varepsilon)$. As $R(\varepsilon)$ is an acAIOSR from
$(\hat S,\hat I,\hat O)$ to $(S,I,O)$ we obtain the inequalities
\begin{IEEEeqnarray*}{rCl}
\hat I_C((\hat x_C,\hat x),(\hat u_C,\hat u))
=
\hat I(\hat x,\hat u)
&\le&
I(x,u)+\gamma_I(\varepsilon')
=I_C((u_C,x),(u_C,u))+\gamma_I(\varepsilon')\\
O_C((x_C,x),(u_C,u))
=
O(x,u)
&\le&
\hat O(\hat x,\hat u)+\gamma_O(\varepsilon')
=
\hat O_C((\hat x_C,\hat x),(\hat u_C,\hat u))+\gamma_O(\varepsilon')
\end{IEEEeqnarray*}
for all $((x_C,x),(\hat x_C,\hat x),(u_C,u),(\hat u_C,\hat u))\in
  R'(\varepsilon)$ and
  $\varepsilon'=\max\{\varepsilon,\sd{}'((u_C,u),(\hat u_C,\hat
  u))\}$.

We apply Theorem~\ref{t:acsim} to $(S_C\times_{R_{C}(\varepsilon)}
S,I_C,O_C)$ and
$(\hat S_C\times_{\hat R_C}\hat S,\hat I_C,\hat O_C)$ with distance function $\sd{}'$ and obtain that
$(S_C\times_{R_{C}(\varepsilon)} S,I_C',O_C)$ is pIODS with the modified input
costs $I_C'((x_C,x),(u_C,u))=\max\{I(x,u),\Gamma(x,u)\}$.
\end{proof}

\begin{remark}\label{r:acsim}
Note that we use Theorem~4 to see that the controlled system
$S_C\times_{R_C(\varepsilon)} S$ is pIODS. If $\gamma_O$ satisfies the
triangle inequality and $\hat I(\hat x,\hat
u)\le I(x,u)$ holds for every $(\hat x,x,\hat u, u)\in R(\varepsilon)$
and $\varepsilon\in \R_{\ge0}$, the premises of Corollary~\ref{c:acsim} are satisfied
and it follows that every behavior $((\xi_C,\xi),(\nu_C,\nu))$ of $S_C\times_{R_C(\varepsilon)} S$
satisfies~\eqref{e:cIODS}.
\end{remark}

\begin{remark}\label{r:main}
Note that the controller $(S_C,R_{C}(\varepsilon))$ for $S_C$ is given by
$S_C=\hat S_C\times_{\hat R_{C}} \hat S$ where
$R_{C}(\varepsilon)$ equals
\mbox{$\{(\hat x_C,\hat x),x,(\hat u_C,\hat u),u)\mid
(\hat x, x,\hat u,u)\in R(\varepsilon)\land (\hat x_C,\hat
x)\in \hat R_{C,X}\}$}, 
see~\eqref{e:rel:transfer}.

Moreover, the parameters $\kappa$, $\beta$ and $\lambda$ and distance
function $\sd{}'$ of the
$(\kappa,\beta,\lambda)$-acIOSR $R'(\varepsilon)$ from
\mbox{$S_C\times_{R_{C}(\varepsilon)} S$} to $S_C$ coincide with
the parameters and distance function $\sd{}$ 
of the $(\kappa,\beta,\lambda)$-acAIOSR from $\hat S$ to $S$ given in
the premise of Theorem~\ref{t:main}.
\end{remark}

\addtocounter{example}{-1}
\begin{example}[DC-DC boost converter (continued)]

Let $(\hat S_C,\hat R_{C})$ denote the controller
from~\cite{GPT10} that renders $D$ positively
invariant with respect to $S_C:=\hat S_C\times_{\hat R_{C}} \hat S$.
Therefore, any behavior $((\hat \xi_C,\hat \xi),(\hat \nu_C,\hat \nu)$ of $\hat S_C\times_{\hat R_{C}} \hat S$  satisfies
$O(\xi_{t},\nu_{t})=I(\xi_{t}, \nu_{t})=0$ and it follows
that $\hat S_C\times_{\hat R_{C}} \hat S$ is $(\hat \gamma,\hat \mu)$-IODS with
$\hat\gamma=0$ and $\hat\mu=0$.  

We apply Theorem~\ref{t:main} and conclude that 
$S_C\times_{R_{C}(\varepsilon)} S$ is pIODS with input costs
$\max\{I,\Gamma\}=|u^d|_2$, since $\Gamma$ induced by $R(\varepsilon)$
and $\sd{}$ is given by $|u^d|_2$. Note that the assumptions of
Corollary~\ref{c:acsim} hold and we can conclude that
any behavior $((\xi_C,\xi),(\nu_{C},\nu))$ of $S_C\times_{R_{C}(\varepsilon)} S$ satisfies
\begin{IEEEeqnarray*}{c}
  |\xi_{t}|_D\le
  \max_{t'\in\intcc{0;t}}\mu_\Delta(\gamma_\Delta(|\nu^d_{t'}|_2),t-t')+\kappa_\Delta
\end{IEEEeqnarray*}
where with $\mu_\Delta(r,t):=(\beta')^t r$, $\gamma_\Delta=1/(\beta'-\beta)$
and $\kappa_\Delta:=\kappa/(1-\beta)$ for some
$\beta'\in\intoo{\beta,1}$. 

The pIODS inequality implies that the system may leave the set $D$
in the presence of disturbances, however in absence of disturbances
the system either stays in $D+\B(\kappa_\Delta)$ or asymptotically
approaches $D+\B(\kappa_\Delta)$. Moreover, contrary to the approach
in~\cite{GPT10} the closed-loop system $S_C\times_{R_C(\varepsilon)}
S$ is non-blocking even in the presents of unbounded disturbances.

\end{example}

Note that in this example, the contraction property of the
system matrices enabled us to establish an acIOASR from the symbolic
model to the concrete system. As a consequence, we could neglect the
continuous disturbances on the symbolic model, but nevertheless
establish the pIODS inequality. We demonstrate in Section~\ref{s:app}
how this procedure  leads to a \emph{separation of concerns} in the
robust controller design for CPS, where a continuous ``low-level'' controller and a
discrete ``high-level'' controller provides robustness with respect to
continuous and discrete disturbances, respectively. In particular, we use a low-level
feedback controller to enforce the contraction property needed to
establish an acIOASR from the symbolic model (without continuous
disturbances) to the concrete CPS. Then we use the synthesis approach
in~\cite{TCRM14} to design a discrete high-level controller that renders the
symbolic model robust against discrete disturbances. Afterwards, we
refine the discrete controller to the concrete CPS according to
Remark~\ref{r:main} and obtain from~Theorem~\ref{t:main} that the
controlled CPS is robust against the continuous as well as discrete
disturbances.

\section{A Compositional Result}
\label{s:composition}

In this section, we show how acASR are preserved under composition.
We analyse four systems $S_1$, $\hat S_1$, $S_2$ and
$\hat S_2$ and assume the existence of the relations
$R_i(\varepsilon)$, $i\in\{1,2\}$ with $R_i(\varepsilon)$ being an
acASR from $\hat S_i$ to $S_i$. Then we show  how
to construct a relation $\hat H$ such that there is an acASR
$R(\varepsilon)$  from $\hat S_1\times_{\hat H} \hat S_2$ to
$S_1\times_H S_2$.

Note that this result is useful to construct symbolic
models that are alternatingly related with
CPS $S_{12}:=S_1\times_H S_2$ that is given by the composition of a system
$S_1$, representing the physical part and system $S_2$, representing the
cyber part. The compositional result enables us to construct a symbolic model 
of the concrete CPS in two steps. In the first
step, we compute symbolic models for the individual parts $S_1$ and
$S_2$. In the second step, we combine those
symbolic models to obtain a symbolic for the composed CPS. Usually,
the cyber part of a CPS is already finite and an abstraction of $S_2$
may not be necessary. In that case, the construction of a symbolic model
of $S_{12}$ is reduced to the computation of symbolic
model for the physical part~$S_1$ using, e.g., the methods presented
in~\cite{PGT08,PT09,GPT10} and \cite[Chapter~11]{Tab09}. We don't provide further details on how to construct such models
here, but refer the reader to Example~\ref{ex:2} and
Section~\ref{s:app} where we illustrate
those approaches with concrete examples.

We begin with the derivation of the compositional result.
Let $S_i$, $\hat S_i$, $i\in\{1,2\}$ be four systems, and let the
relations $R_i(\varepsilon)$ be acASR from $\hat S_i$ to $S_i$.
Suppose we are given $H\subseteq X_1\times X_2\times U_1\times U_2$,
then we define the relation $\hat H\subseteq \hat X_1\times \hat X_2\times
\hat U_1\times \hat U_2$ by
\begin{IEEEeqnarray}{c}\label{e:interconnection}
\{(\hat x_1,\hat x_2,\hat u_1,\hat u_2)\mid
\exists \varepsilon,x_i,u_i:(\hat x_i,x_i,\hat  u_i, u_i)\in
R_i(\varepsilon),i\in\{1,2\}\land (x_1,x_2,u_1,u_2)\in H\}
\end{IEEEeqnarray}
and $R(\varepsilon)\subseteq \hat X_{12}\times X_{12}\times
\hat U_{12}\times U_{12}$ by
\begin{IEEEeqnarray}{c}\label{e:compASR}
\{(\hat x_{12},x_{12},\hat  u_{12}, u_{12})\mid
(\hat x_i,x_i, \hat u_i,u_i)\in R_i(\varepsilon),i\in\{1,2\}\}
\end{IEEEeqnarray}
We use the following assumption
\begin{multline}\label{e:comp}
(\hat x_i,x_i)\in R_{i,X}(\varepsilon),i\in\{1,2\}\land (x_1,x_2)\in
H_X\implies\\ \exists u_i,\exists \hat u_i:(\hat x_i,x_i,\hat u_i,u_i)\in R_i(\varepsilon)\land (x_1,x_2,u_1,u_2)\in H
\end{multline}
Intuitively, we ensure with this assumption that if $(x_1,x_2)\in H_X$
and the states $\hat x_i$ are related to $x_i$ for $i\in\{1,2\}$ then
$(\hat x_1,\hat x_2)\in \hat H_X$.

\begin{theorem}\label{t:comp}
Let $S_i$, $\hat S_i$, $i\in\{1,2\}$ be four systems, and let the
relations $R_i(\varepsilon)$ be $(\kappa_i,\beta_i,\lambda_i)$-acASR
from $\hat S_i$ to $S_i$ with distance function~$\sd{i}$.
Let $H\subseteq X_1\times X_2\times U_1\times U_2$ be a relation and
$\hat H\subseteq \hat X_1\times \hat X_2\times \hat U_1\times \hat
U_2$ be obtained from~\eqref{e:interconnection}. If~\eqref{e:comp} holds, then $R(\varepsilon)$ as defined
in~\eqref{e:compASR} is an $(\kappa,\beta,\lambda)$-acASR from $\hat
S_1\times_{\hat H}\hat S_2$ to $S_1\times_H S_2$ with
$\kappa:=\max_i\{\kappa_i\}$, 
$\beta:=\max_i\{\beta_i\}$ and
$\lambda:=\max_i\{\lambda_i\}$ with distance function
$\sd{12}(u_{12},\hat u_{12}):=\max_i\{\sd{i}(u_i,\hat u_i)\}$.
\end{theorem}
\begin{proof}
The property $R(\varepsilon)\subseteq R(\varepsilon')$
whenever $\varepsilon\le \varepsilon'$ is directly inherited from
$R_i(\varepsilon)$. Let $\hat x_{12}\in \hat X_{120}$ which
implies $\hat x_i\in \hat X_{i0}$, $i\in\{1,2\}$. Therefore, there exist $x_i\in
X_{i0}$ with $(\hat x_i,x_i)\in R_{i,X}(\kappa_i)$ and thereby we have
$(\hat x_{12},x_{12})\in R_X(\kappa)$.

Consider $(\hat x_{12},x_{12})\in R_X(\varepsilon)$ and $\hat u^c_{12}\in
\hat U^c_{12}(\hat x_{12})$. This implies
\begin{IEEEeqnarray}{t}
$(\hat x_i, x_i)\in
R_{i,X}(\varepsilon)$ and $\hat u^c_i\in \hat U_i^c(\hat x_i)$ for
$i\in\{1,2\}$. \label{e:h:comp:a}
\end{IEEEeqnarray}
By \eqref{e:h:comp:a}, we can pick $u^c_i\in U^c_i$ such that the
tuple $(\hat x_i, x_i,\hat u^c_i, u^c_i)$
satisfies 2.a) in Definition~\ref{d:acASR}. Let $u^d_{i}\in U^d_{i}$ and
$x'_{12}\in r_{12}(x_{12},u_{12})$ where $u_i=(u^c_i,u^d_i)$ and  $u_{12}=(u_1,u_2)$. By our choice of
$\hat u^c_i$ there exist $\hat u^d_i\in \hat U^d_i$ such that
$(\hat x_i,x_i, \hat u_i,u_i)\in
R_i(\varepsilon)$ with  $\hat u_i=(\hat
u^c_i,\hat u^d_i)$ and it follows that $(\hat x_{12}, x_{12},
\hat u_{12},u_{12})\in R(\varepsilon)$ where
$\hat u_{12}=(\hat u_1,\hat u_2)$.

For $i\in\{1,2\}$, we choose $\hat x'_i\in \hat r_i(\hat x_i,\hat u_i)$ such that $(\hat x_i',x_i')\in R_{i,X}(\varepsilon_i')$ with
$\varepsilon'_i=\kappa_i+\beta_i\varepsilon+\lambda_i\sd{i}(\hat u_i,u_i)$. It remains to show that  $\hat x'_{12}\in \hat
r_{12}(\hat x_{12},\hat u_{12})$ from which follows
that $(\hat x'_{12},x'_{12})\in R_X(\varepsilon')$ with
$\varepsilon'=\kappa+\beta\varepsilon+\lambda\sd{}(\hat u_{12},u_{12})$. We need to check 4.c) in
Definition~\ref{d:comp}.

Since $(\hat x_i,x_i,\hat u_i,u_i)\in R_i(\varepsilon)$ and
$(x_1,x_2,u_1,u_2)\in H$ we have
 $(\hat x_1,\hat x_2,\hat u_1,\hat u_2)\in \hat
H$ and it remains to show that $\hat x'_{12}\in \hat H_X$. That
follows by~\eqref{e:compASR}, since we know that $(\hat x'_i, x'_i)\in
R_{i,X}(\varepsilon')$ and $(x'_1,x'_2)\in H_X$.
\end{proof}

\section{A Mobile Robot Example}
\label{s:app}

In this section, we demonstrate our results in terms of a simple example with a robot
moving in the plane equipped with an omnidirectional drive. We model
the sampled dynamics of the robot by the difference equation
\begin{IEEEeqnarray*}{c}
  \xi_{t+1}=\xi_{t}+\nu_{t}
\end{IEEEeqnarray*}
where $\xi_t\in\R^2$ is the position of the robot and
$\nu_t\in\R^2$ is the control input.
We assume that the
control signal is sent to the mobile robot over a wireless
communication channel with possible package dropouts. We apply
the presented abstraction and refinement framework to design a robust
controller for the robot over the lossy channel.  As a
first step, we construct a symbolic model that alternatingly
simulates the robot. Here we use Theorem~\ref{t:comp} to construct symbolic models of the physical part and cyber
part individually and then compose those models to obtain a
symbolic model of the overall robot with communication channel.
Afterwards, we use the approach from~\cite{TCRM14}
to synthesize a robust controller for the symbolic model. Finally, we
apply Theorem~\ref{t:main} to refine the controller for the symbolic
model to the robot.

{\bf The system description.}
We assume that the robot drive is equipped with low-level controllers that
we use to enforce the sampled-data dynamics
\begin{IEEEeqnarray}{c'c}\label{e:sys}
  \xi_{t+1}=0.8\xi_{t}+\nu_{t}+\omega_{t}.
\end{IEEEeqnarray}
We use $\omega_t\in\R^2$ to model
actuator errors and/or sensor noise. A real-world example of a robot
that fits our assumptions is \emph{Robotino},
see~\cite{WB10}.
We cast
\eqref{e:sys} as the system
$S_1=(X_1,X_{10},U_1,r_1)$ with $X_1=\R^2$,
$X_{10}=\{x_{10}\}$, $U_1=U^c_1\times U^d_1$, $U^c_1=U^d_1=\R^2$ and $r_1$ is defined in
the obvious way. 

Moreover, we assume that the high-level control signal $u$ is sent to the
actuator via a wireless connection where package dropouts might occur.
However, for simplicity of the presentation, we assume that two
packages are never dropped consecutively. 
We use the system
$S_2=(X_2,X_{20},U_2,r_2)$ with $X_2=\{a_0,a_1\}$, $X_{20}=X_2$ and
$U_2=U^d_2= D$ and $D=\{\perp,\top\}$
to model that behavior. The
dynamics $r_2$ of the system $S_2$ is illustrated in Figure~\ref{f:actuator}.
\begin{figure}[htb]
\centering
\begin{tikzpicture}[->,shorten >=1pt,thick,node distance=2cm,auto]
\tikzstyle{every state}=[scale=0.8,draw=black!100,fill=black!10]
\node[state]  (00)                 {$a_{0}$};
\node[state]  (11) [right of=00]   {$a_{1}$};
\path (00) edge [bend left]  node {\footnotesize $\top$} (11)
      (00) edge [loop above] node[left=3pt] {\footnotesize $\bot$} (00)
      (11) edge [bend left]  node
      {\footnotesize$\perp,\top $} (00);
\end{tikzpicture}
\caption{Dynamics to model possible package dropouts.}\label{f:actuator}
\end{figure}
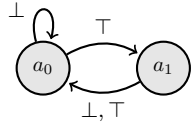
Our model of the wireless communication acts like a switch with respect to the control input
$\bar u\in\R^2$. If a package dropout occurs, i.e.,
$x_2=a_1$, we apply zero as control
input $u=0$. If no dropout occurs, i.e., $x_2=a_0$, the control input is
$u=\bar u$ since the robot successfully received a
control update. The transition between the nominal state $x_2=a_0$ and
the state when a package dropout occurs $x_2=a_1$ is modelled by the
perturbation signal $\top$. The continuation of the nominal behavior,
i.e., no package dropout occurs is modelled by the nominal input
$\bot$. 

We define the composed system
$S_{12}:=S_1\times_H S_2$ using the relation $H\subseteq X_1\times X_2\times U_1\times U_2$
which is implicitly given by
\begin{IEEEeqnarray*}{c}
(x_1,x_{2},(u^c_1,u^d_1),u_2)\in H:\iff (x_2=a_1\implies u^c_1=0).
\end{IEEEeqnarray*}
In this way only the zero control input $u^c_1=0$ is allowed when the system
$S_2$ is in state $x_2=a_1$.

We would like to enforce a periodic behavior
which we express as a cycle along the states displayed in
Figure~\ref{f:1}.
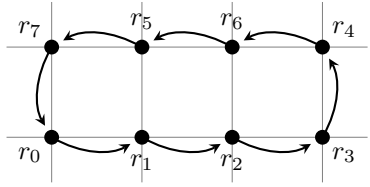
\begin{figure}[htb]
\centering
\begin{tikzpicture}[thick,scale=1.2]

\draw[gray, very thin, step=1] (-.5,-.5) grid (3.5,1.5);

\draw[fill,color=black!100!white]
                  (0,0)  circle (2pt) node[anchor=north east] {$r_0$}
                  (1,0)  circle (2pt) node[below=3pt]         {$r_1$}
                  (2,0)  circle (2pt) node[below=3pt]         {$r_2$}
                  (3,0)  circle (2pt) node[anchor=north west] {$r_3$}
                  (0,1)  circle (2pt) node[anchor=south east] {$r_7$}
                  (1,1)  circle (2pt) node[above=3pt]         {$r_5$}
                  (2,1)  circle (2pt) node[above=3pt]         {$r_6$}
                  (3,1)  circle (2pt) node[anchor=south west] {$r_4$};
\path[->,>=stealth,shorten >=5pt]
           (0,0) edge[bend right] (1,0)
           (1,0) edge[bend right] (2,0)
           (2,0) edge[bend right] (3,0)
           (3,0) edge[bend right] (3,1)
           (3,1) edge[bend right] (2,1)
           (2,1) edge[bend right] (1,1)
           (1,1) edge[bend right] (0,1)
           (0,1) edge[bend right] (0,0);

\end{tikzpicture}
\caption{Desired trajectory in the state space.}\label{f:1}
\end{figure}
In order to express our desired behavior in terms of the output
costs, we introduce a system $S_3=(X_3,\{r_0\},U_3,r_3)$ with
$X_3=\{r_i\}$, $i\in\{0,\ldots,7\}$, $X_{30}=\{r_0\}$, $U_3=\{\epsilon\}$ and 
$r_3(x_3,u_3)$ given according to~Figure~\ref{f:1}. 
The reference states  $r_i\in\R^2$ are given by 
\begin{IEEEeqnarray*}{c,c,c,c,c,c,c,c}
r_{0}=\begin{bmatrix} 0, \, 0\end{bmatrix}^\top,&
r_{1}=\begin{bmatrix} 1, \, 0\end{bmatrix}^\top,&
r_{2}=\begin{bmatrix} 2, \, 0\end{bmatrix}^\top,&
r_{3}=\begin{bmatrix} 3, \, 0\end{bmatrix}^\top,\\
r_{4}=\begin{bmatrix} 3, \, 1\end{bmatrix}^\top,&
r_{5}=\begin{bmatrix} 2, \, 1\end{bmatrix}^\top,&
r_{6}=\begin{bmatrix} 1, \, 1\end{bmatrix}^\top,&
r_{7}=\begin{bmatrix} 0, \, 1\end{bmatrix}^\top.
\end{IEEEeqnarray*}
The overall system is obtained as the composition of the three systems
$S_{123}=S_{12}\times_G S_3$ with respect to $G:=X_{12} \times X_3\times
U_{12}\times U_3$. We define the output costs $O:X_1\times X_3\to \R_{\ge0}$ by
\begin{IEEEeqnarray*}{c}
  O(x_1,x_3):= |x_1-x_3|
\end{IEEEeqnarray*}
and choose the input costs $I:X_2\times U^d_1\to \R_{\ge0}$ simply as
\begin{IEEEeqnarray*}{c}
  I(x_2,u_1^d):= I_d(x_2)+|u_1^d|,
\end{IEEEeqnarray*}
with $I_d(a_0):=0$ and $I_d(a_1):=1$. Note that we omit the
independent variables in $O$ and $I$. The value of the output costs indicates how
well the robot is following the nominal behavior. The costs are zero, if
the robot follows the system $S_3$ and non-zero otherwise. The input
costs are used to quantify the possible disturbances.

{\bf The symbolic model.}
We continue with the construction of the symbolic model $\hat S_{123}$
for~$S_{123}$, where we construct symbolic models $\hat S_i$,
  $i\in\{1,2,3\}$ for each
subsystem $S_i$, respectively, and then use Theorem~\ref{t:comp} to
compose the individual models $\hat S_i$ to $\hat S_{123}$.

First we introduce the symbolic model $\hat S_1$ of $S_1$ based on a
discretization of 
the state space and input space of $S_1$. We choose
$\hat X_1=[\intcc{-1,4}^2]_\kappa$, $\hat
U^c_1=[\intcc{-3,3}^2]_\kappa$ and $\hat U^d_1=\{[0,0]^\top\}$. 
Note that 
we neglect the disturbances $U^d_1$ on the symbolic model $\hat S_1$.
We set the discretization parameter to~$\kappa=0.05$.
We leave it to the reader to check that the relation
$R_1(\varepsilon)\subseteq \hat X_1\times X_1\times\hat  U_1\times U_1$ given by
\begin{IEEEeqnarray*}{c}
\{(\hat x_1,x_1,(\hat u^c_1,0),(u^c_1,u^d_1))\mid |x_1-\hat x_1|\le
\varepsilon\land u^c_1=\hat u^c_1\}
\end{IEEEeqnarray*}
is an $(0.05,0.8,1)$-acASR from $\hat S_1$ to $S_1$ with distance
function $\sd{1}((\hat u_1^c,0),(u_1^c,u_1^d))=|u^d_1|$.

The symbolic models for $S_2$ and $S_3$ are directly given by $\hat
S_2=S_2$ and $\hat S_3=S_3$ since $S_2$ and $S_3$ are finite. It is
straightforward to see that the relations \mbox{$R_i:=\{(\hat
x_i,x_i,\hat u_i,u_i)\mid \hat x_i=x_i\land \hat u_i=u_i\}$}, $i\in\{2,3\}$ are 
$(0,0,0)$-acASR from $\hat S_i$ to $S_i$ with distance functions
$\sd{i}(\hat u_i,u_i)=0$.

Now we apply Theorem~\ref{t:comp} to see that $R_{12}(\varepsilon)\subseteq
\hat X_{12}\times X_{12}\times  \hat U_{12}\times U_{12}$ given by 
\begin{IEEEeqnarray*}{c}
\{(\hat x_{12},x_{12},\hat u_{12},u_{12})\mid (\hat x_1, x_1,\hat u_1,u_1)\in R_1(\varepsilon)\land x_2=\hat x_2\land u_2=\hat u_2\}
\end{IEEEeqnarray*}
is an $(0.05,0.8,1)$-acASR from $\hat S_{12}:=\hat S_1\times_{\hat H} \hat S_2$ to
$S_{12}$ with distance function $\sd{12}(\hat u_{12},u_{12})=\sd{1}(\hat u_1,u_1)=|u^d_1|$.
The relation $\hat H\subseteq \hat X_1\times \hat X_2\times \hat U_1\times\hat  U_2$ results
from~\eqref{e:interconnection} to
$(\hat x_{1},\hat x_{2},\hat u_{1},\hat u_{2})\in \hat H$ iff $(\hat x_2=a_1\implies
\hat u^c_1=0)$.
By the same arguments we see that the relation
\begin{IEEEeqnarray*}{c}
R_{123}(\varepsilon):=\{(\hat x_{123},x_{123},\hat u_{123},
u_{123})\mid (\hat x_{12},x_{12},\hat u_{12}, u_{12})\in
R_{12}(\varepsilon)\land \hat x_3=x_3\}
\end{IEEEeqnarray*}
is an $(0.05,0.8,1)$-acASR
from $\hat S_{123}:=\hat S_{12}\times_{\hat G} \hat S_3$ to $S_{123}$
with distance functions $\sd{123}(\hat u_{123},u_{123})=\sd{1}(\hat
u_1,u_1)$, where $\hat G:=\hat X_{12}\times \hat X_3\times \hat
U_{12}\times \hat U_3$.

We choose the cost functions $\hat I$ and $\hat O$ for
$\hat S_{123}$ to be $\hat I(\hat x_2):=I_d(\hat x_2)$ and $\hat
O(\hat x_1, \hat x_3):=|\hat x_1-\hat x_3|_\kappa$. We
remark that the cost functions satisfy~\eqref{e:costASR} with 
$\gamma_I= 0$ and $\gamma_O(c)=c+\min\{c,\kappa\}$ and thereby follows
that $R_{123}(\varepsilon)$ is an acAIOSR from $\hat S_{123}$ to
$S_{123}$.

We use the synthesis approach in~\cite{TCRM14} to compute a
controller $(\hat S_C,\hat R_{C})$ that renders the system $\hat
S_{123}$ IODS. As a result, we obtain  the IODS inequality
\begin{IEEEeqnarray}{c}\label{e:robot:iods}
|\hat \xi_{1,t}-\hat \xi_{3,t}|_\kappa\le
\max_{t'\in\intcc{0;t}}\{1.4 I_d(\hat \xi_{t'}))-1.4(t-t')\}
\end{IEEEeqnarray}
for every behavior $(\hat \xi,\hat \nu)$ of the controlled
system $S_{C}:=\hat S_C\times_{\hat R_C}\hat S_{123}$.
Note that with $\gamma=\eta=1.4$ the effect of the
disturbance $\hat x_2=a_1$ at time $t$ disappears after one step.

{\bf Controller refinement.}  We now apply Theorem 5 to refine the
controller for $\hat S_{123}$ to a controller for $S_{123}$.
First, note that $R_{123}(\varepsilon)$ is a $(0.05,0.8,1)$-acAIOSR
from $\hat S_{123}$ to $S_{123}$ with $\Gamma(x_{123},
((u_1^c,u_1^d),u_{23}))=|u^d_1|$ that satisfies~\eqref{e:infR} and
$R_{123}(\varepsilon)$ satisfies~\eqref{e:tightASR}. Moreover, \mbox{$(\hat
S_C\times_{\hat R_C} \hat S, \hat I,\hat O)$} is IODS with the
inequality~\eqref{e:robot:iods}.
As a consequence there exists a controller $(S_C,R_C(\varepsilon))$ for
$S_{123}$ and  the controlled system is pIODS.
Furthermore, since $\hat I(\hat x,\hat u)\le I(x,u)$ for all related
tuples $(\hat x, x,\hat u, u)$ we can apply Corollary~\ref{c:acsim}
and the inequality 
\begin{IEEEeqnarray*}{C}
\begin{IEEEeqnarraybox}[][c]{rCl}
O(\xi_t)
&\le&
\max_{t'\in\intcc{0;t}}\{1.4 I_d(\xi_t)-1.4(t-t')\}+
\max_{t'\in\intcc{0,t}}\tfrac{1}{\beta'-0.8}(\beta')^{t-t'}|\pi_{U^d_1}(\nu_{t'})|
+
0.25
\end{IEEEeqnarraybox}
\end{IEEEeqnarray*}
follows for any behavior $(\xi,\nu)$ of $S_{C}\times_{R_C(\varepsilon)}S$ and
  any $\beta'\in\intoo{0.8,1}$. 

This example demonstrates nicely how our results enable us to separate
the design procedure to establish robustness with respect to
continuous and discrete disturbances.  We used the low-level
controllers of the robot to enforce the contractive
dynamics~\eqref{e:sys} so that $S$ admits an acAIOSR. We used the
discrete design procedure~\cite{TCRM14} to establish the IODS
inequality~\eqref{e:robot:iods} for the symbolic model with respect to
the discrete disturbances. As the previous pIODS inequality shows, the
final controlled system is robust with respect to both discrete as
well as continuous disturbances.

\bibliographystyle{plain}
\bibliography{ios}

\appendix

\begin{lemma}\label{l:kldbound}
For every $\mu\in{\mathcal{KLD}}$, $\gamma\in{\mathcal K}$ and
$c'\in\R_{\ge 0}$ the following inequality 
holds
\begin{IEEEeqnarray*}{C}
  \mu(\gamma(c+c'),t)\le \mu(\gamma'(c),t)+\sigma(c')
\end{IEEEeqnarray*}
for all $c,t\in\R_{\ge0}$ with $\gamma'(c)=2\gamma(2c)$ and $\sigma(c')=\mu(\gamma'(c'),0)$.
\end{lemma}
\begin{proof}[Proof of Lemma~\ref{l:kldbound}]
We apply the fact $a+b\le \max\{2a,2b\}$ twice. First, for $\gamma$ we
get $\gamma(c+c')\le\gamma(\max\{2c,2c'\})\le \gamma(2c)+\gamma(2c')$.
Then for $\mu$ we obtain
\begin{IEEEeqnarray*}{rCl+r*}
\mu(\gamma(c+c'),t)
&\le&
\mu(\gamma(2c)+\gamma(2c'),t)\\
&\le&
\mu(\max\{2\gamma(2c),2\gamma(2c')\},t)\\
&\le&
\mu(2\gamma(2c),t)+ \mu(2\gamma(2c'),0).&\qedhere
\end{IEEEeqnarray*}
\end{proof}

\begin{lemma}\label{l:transformation}
Suppose we are given $\gamma\in{\mathcal K}$ and $\mu\in{\mathcal{KLD}}$. Then
there exists $\mu'\in{\mathcal{KLD}}$ such that 
\begin{IEEEeqnarray}{c'c}\label{e:transformation}
  \gamma(\mu(c,t))=\mu'(\gamma(c),t)
\end{IEEEeqnarray}
holds for all $c\in\R_{\ge0}$ and $t\in\N$.
\end{lemma}
\begin{proof}
We define $\mu': \R_{\ge0}\times\N\to \R_{\ge0}$ iteratively by
\begin{IEEEeqnarray*}{c'c}
  \mu'(c,0)\defeq r, &\mu'(c,t+1)\defeq g(\mu'(c,t))
\end{IEEEeqnarray*}
for all $c\in\R_{\ge0}$ and $t\in\N$, where $g(c)\defeq
  \gamma(\mu(\gamma^{-1}(c),1))$.  It is easy
  to see by induction over $t\in\N$ that $\mu'$ 
satisfies~\eqref{e:transformation}. Hence, $\mu'$ is a ${\mathcal{KL}}$
function and by the iterative definition follows that $\mu'\in
{\mathcal{KLD}}$.
\end{proof}
\begin{lemma}\label{l:maxbound}
Suppose we are given $\mu_a,\mu_b\in{\mathcal{KLD}}$. Then
there exists $\mu\in{\mathcal{KLD}}$ such that 
\begin{IEEEeqnarray}{c'c}\label{e:maxbound}
  \max_{t'\in\intcc{0;t}}\mu_a(c,t')+\max_{t'\in\intcc{0;t}}\mu_b(c,t')\le\max_{t'\in\intcc{0;t}}\mu(2c,t')
\end{IEEEeqnarray}
holds for all $c\in\R_{\ge0}$ and $t\in\N$.
\end{lemma}
\begin{proof}
First let us remark that 
\begin{IEEEeqnarray}{c'c}\label{e:h:maxbound}
  \max_{t'\in\intcc{0;t}}\mu_a(c,t')+\max_{t'\in\intcc{0;t}}\mu_b(c,t')
 \\ \le
  \max\{2\max_{t'\in\intcc{0;t}}\mu_a(c,t'),2\max_{t'\in\intcc{0;t}}\mu_b(c,t')\} \end{IEEEeqnarray}
holds for all $c\in\R_{\ge0}$ and $t\in\N$. Now we define
$\mu:\R_{\ge0}\times \N\to\R_{\ge0}$ recursively for all
$c\in\R_{\ge0}$ and $t\in\N$ by $\mu(c,0)\defeq c$, $\mu(c,t+1)\defeq g(\mu(c,t))$
with $g(c)\defeq\max\{2\mu_a(c,1),2\mu_b(c,1)\}$. To show
\eqref{e:maxbound}, in view of \eqref{e:h:maxbound}, it suffices to
show that $2\mu_j(c,t)\le \mu(c,t)$
holds for all $j\in\{a,b\}$, $c\in\R_{\ge0}$ and $t\in\N$. We fix
$j\in\{a,b\}$ and $c\in \R_{\ge0}$ and proceed
by induction over $t\in\N$. The base case $t=0$ is trivial. Suppose
the induction hypothesis holds, then we derive
$2\mu_j(c,t+1)
=
2\mu_j(\mu_j(c,t),1)
\le 2\mu_j(\mu(2c,t),1)
\le g(\mu(2c,t))=\mu(2c,t+1).$
\end{proof}

\begin{proof}[Proof of Lemma~\ref{l:transfer}]
We show that the relation $R_{123}(\varepsilon)$ defined by
\begin{multline}\label{e:rel:transfer}
\{(x_{12},x_3,u_{12},u_3)
\in
X_{12}\times X_{3}\times U_{12}\times U_{3})
\mid 
(x_2, x_3,u_2,u_3)\in  R_{23}(\varepsilon)
\land
(x_1, x_2)\in R_{12,X}\}
\end{multline}
fulfills the claim of the lemma.

First let us note that $R_{123}(\varepsilon')\subseteq R_{123}(\varepsilon)$
whenever $\varepsilon'\le\varepsilon$ is directly inherited from the
inclusion $R_{23}(\varepsilon')\subseteq R_{23}(\varepsilon)$ for $\varepsilon'\le
\varepsilon$. Moreover, $e(x_{12},x_{3})=e(x_2,x_3)$ whenever
$(x_1,x_2)\in R_{12,X}$ which implies
that $(x_2,x_3)\in R_{23,X}(e(x_2, x_3))$ and $(x_1, x_2)\in
R_{12,X}$
whenever $e(x_{12}, x_{3})<\infty$. Hence, $R_{123}(\varepsilon)$
satisfies~\eqref{e:infR}.

We proceed by checking 1) of Def.~\ref{d:acASR}. Let $ x_{12}\in
 X_{120}\subseteq R_{12,X}$. Since for every $x_2\in X_{20}$ there
is $x_3\in X_{30}$ with $(x_2, x_3)\in  R_{23,X}(\kappa)$,
there exists \mbox{$x_3\in X_{30}$} with
\mbox{$(x_{12},x_3)\in R_{123,X}(\kappa)$}.

Let us now check 2) of Def.~\ref{d:acASR}. Let 
\mbox{$(x_{12}, x_{3})\in R_{123,X}(\varepsilon)$} and $u_{12}^c\in U_{12}^c( x_{12})$. 
This implies: 
\begin{itemize}
\item[a)] $(x_2, x_3)\in R_{23,X}(\varepsilon)$ and $(x_1, x_2)\in R_{12,X}$;
\item[b)] $r_{12}(x_{12},(u^c_{12},u^d_{12}))\neq\emptyset$ for any $u^d_{12}\in U^d_{12}$.
\end{itemize}

Since $(x_2, x_3)\in R_{23,X}(\varepsilon)$ and $u^c_2\in
U^c_2(x_2)$ we
can choose $u_3^c\in U_3^c$ so that 2.a) of Def.~\ref{d:acASR}
holds. Now for $u^d_3\in U^d_3$ and $x_3'\in r_3(x_3,u_3)$
we can pick $u^d_2\in U^d_2$ and $ x'_2\in r_2(x_2, u_2)$ such
that $(x_2,x_3,u_2,u_3)\in R_{23}(\varepsilon)$ and
$(x'_2, x'_3)\in
R_{23}(\varepsilon')$ with
$\varepsilon'=\kappa+\beta\varepsilon+\lambda\sd{23}(u_2, u_3)$.

Moreover, from b) and~\eqref{e:tightASR} follows that 
$(x_1,x_2,u^c_1,u^c_2)$ satisfy 2.a) of
Def.~\ref{d:acASR}. Therefore,
there exist \mbox{$ u^d_1\in U^d_1$} and
$ x_1'\in  r_1( x_1, u_1)$
for our choice of $u^d_2$ and $ x_2'$ so that 
\mbox{$(x_1, x_2, u_1, u_2)\in  R_{12}$} and
$(x'_1, x_2')\in R_{12,X}$.

In the previous two paragraphs we showed
$(x_{12},x_3,,u_{12}, u_{3})\in R_{123}(\varepsilon)$ and
$(x'_{12}, x'_{3})\in
R_{123,X}(\varepsilon')$ which implies that $R_{123}(\varepsilon)$ is a
$(\kappa,\beta,\lambda)$-acASR from $S_{12}$ to $S_3$
with the distance function given by
$\sd{{123}}(u_{12},u_3)=\sd{{23}}(u_2,u_3)$.
\end{proof}

\begin{proof}[Proof of Lemma~\ref{l:controllerNB}]
We only need to show that $S_{12}=S_1\times_{R_{12}(\varepsilon)} S_2$ is
non-blocking as defined in~Def.~\ref{d:controller}.
By definition of $S_{12}$ and \eqref{e:infR} every reachable state
$x_{12}$
of $S_{12}$ satisfies
$(x_1, x_2)\in R_{12,X}(e(x_1, x_2))$. Now it is easy to check 
with the help 2.a) in the
Def.~\ref{d:acASR} that $S_{12}$ satisfies the non-blocking condition.
\end{proof}

\begin{proof}[Proof of Lemma~\ref{l:closedloop}]
We leave it to reader to check that 
the relation
$R_{121}(\varepsilon)\subseteq X_{12}\times X_1\times U_{12}\times
U_1$ given by
\begin{IEEEeqnarray*}{c}
\{
  (x_{12},x'_1,u_{12},u'_1)
  \mid
  (x_1,x_2,u_1,u_2)\in R_{12}(\varepsilon)
  \land
  x_1=x'_1
  \land 
  u_1=u'_1
\}
\end{IEEEeqnarray*}
is an acSR from $S_{12}$ to $S_1$.
\end{proof}

\end{document}